\documentclass[onecolumn, 11pt]{article}

\pdfoutput=1

\usepackage[utf8]{inputenc}

\usepackage{longtable}
\usepackage[left=2cm,right=2cm,top=3cm,bottom=2.5cm]{geometry}
\usepackage{media9} 
\usepackage{fancyhdr}
\usepackage{setspace}
\usepackage{import}
\usepackage[hang, small,labelfont=bf,up,textfont=it,up]{caption}
\usepackage{subcaption}

\usepackage{sectsty}
\sectionfont{\centering}
\usepackage{float}

\usepackage{blkarray}
\usepackage{graphicx}
\usepackage{mwe}

\newcommand{\imineq}[2]{\vcenter{\hbox{\includegraphics[scale=0.3]{network001}}}}
\newcommand{\motif}[2]{\vcenter{\hbox{\includegraphics[scale=0.3]{motif001}}}}

\usepackage{todonotes}
\usepackage[english]{babel}										
\usepackage[protrusion=true,expansion=true]{microtype}
\usepackage{gensymb}
\usepackage{booktabs}
\usepackage{mathptmx}
\usepackage{siunitx}
\usepackage{subcaption}
\usepackage{nomencl}
\usepackage{multirow}
\usepackage{framed}
\usepackage{multicol}
\usepackage{url}
\usepackage{amsmath,amsfonts,amsthm,amssymb}

\usepackage{pdflscape}

\usepackage{color}
\makeatletter
\def\BState{\State\hskip-\ALG@thistlm}
\makeatother

\usepackage{verbatim}

\usepackage{tablefootnote}

\theoremstyle{definition}

\usepackage[explicit]{titlesec}
\titleformat{name=\paragraph,numberless}[runin]
  {\normalfont\normalsize\bfseries}{}{15pt}{\uline{#1.}}
  
\usepackage[toc,page]{appendix}

\usepackage{enumerate}

\usepackage{algorithm}
\usepackage{algpseudocode}

\newtheorem{remark}{Remark}[section]

\usepackage{footnote}
\makesavenoteenv{tabular}
\makesavenoteenv{table}

\usepackage{threeparttable}

\usepackage{siunitx}
\usepackage{dcolumn}
\newcolumntype{d}[1]{D{.}{.}{#1}}
\usepackage{etoolbox}
\robustify\bfseries
\usepackage{chngcntr}

\usepackage{array}

\usepackage[numbers]{natbib}

\theoremstyle{plain}
\newtheorem{theorem}{Theorem}[section]

\newtheorem{proposition}[theorem]{Proposition}

\newcommand{\tr}{\operatorname{tr}}

\usepackage{enumitem}

\usepackage{mathtools}

\usepackage{sectsty}
\sectionfont{\bfseries\Large\raggedright}

\begin{document}

\singlespace

\title{Cliques and a New Measure of Clustering:\\
with Application to U.S. Domestic Airlines
}

\author{
{Steve Lawford$^\dagger$ and Yll Mehmeti}\\[2pt]
Data, Economics and Interactive Visualization (DEVI) group, ENAC (University of Toulouse),\\7 avenue Edouard Belin, CS 54005, 31055, Toulouse, Cedex 4, France\\
$^\dagger$Corresponding author. Email: steve.lawford@enac.fr
}

\date{}

\singlespace

\maketitle

\begin{abstract}

    \noindent We propose a higher-order generalization of the well-known overall clustering coefficient for triples $C(3)$ to any number of nodes. We give analytic formulae for the special cases of three, four, and five nodes and show that they have very fast runtime performance for small graphs. We discuss some theoretical properties and limitations of the new measure, and use it to provide insight into dynamic changes in the structure of U.S. airline networks.
    
\end{abstract}

\section{Introduction}
Complex networks are widely used to describe important systems, with applications to biology, technology and infrastructure, and social and economic relationships \cite{albert_barabasi02, amaral_ottino04, costa_etal07, newman03, strogatz01, zou_etal19}. A network or ``graph'' involves a set of nodes or ``vertices'' that are linked by edges. For example, an airline company's transportation of passengers can be thought of as a network of airports (nodes) joined by routes that have regular service (edges). The statistical physics and graph theory communities have focused in particular on the topology and dynamics of random and real-world networks, and have been successful in identifying robust structural features and organizational principles.\footnote{The field is continually expanding and is far too large to survey here. We thank one anonymous referee for pointing us towards recent work on edge prediction \cite{zhu_etal14} and multiplex models \cite{zhu_etal18, zhu_etal19b}.} These include the small-world property, characterized by systems that are highly clustered but have short characteristic path lengths; and scale-free networks, which means that the number of neighbours of a node, or its ``degree'', follows a power-law distribution whereby the topology of the system is dominated by a few high degree nodes \cite{barabasi_albert99, cimini_etal19, watts_strogatz98}.

One common property of networks is clustering or ``transitivity'', which measures the relative frequency with which two neighbours of a given node are also neighbours of one another, forming a connected triangle of nodes. Many real-world networks display higher levels of clustering than would be expected if those networks were random, with nodes creating tightly connected groups \cite{albert_barabasi02, newman01c, strogatz01, watts_strogatz98}. Clustering is especially important in economic and social networks, and there is strong evidence that it is related to cooperative social behaviour and beneficial information and reputation transfer \cite{jackson08, jackson14, jackson_etal17, newman03}. Other recent examples of the empirical application of graphs in economics include \cite{banerjee_etal13,faris_felmlee11,jackson14} (social networks) and \cite{akbas_etal16,cohen-cole_etal14,el-khatib_etal15,hochberg_etal07,robinson_stuart04} (financial networks).

In this paper, we focus our empirical application on air transportation. Recent work has considered air cargo networks \cite{bombelli_etal20, malighetti_etal19}, the world-wide airport network \cite{cheung_etal20, guimera_amaral04, guimera_etal05, lordan_sallan19, ryczkowski_etal17, verma_etal14}, and airline networks in the U.S. \cite{aguirregabiria_ho12, baumgarten_etal14, gautreau_etal09, lin_ban14, roucolle_etal20, roucolle_etal20b, wuellner_etal10}, Europe \cite{reggiani_etal10}, and China \cite{chen_etal20, du_etal16}; good surveys of research in this area appear in \cite{lordan_etal14, roucolle_etal20, zanin_lillo13}. Typically, these papers report a selection of summary statistics to capture global or local aspects of the network, and provide insight into topology and dynamics that would not be available from  other methods.

One widely used measure of clustering is the \emph{overall clustering coefficient} or ``transitivity'' which is defined in \cite{barrat_weigt00, newman03, newman03d, newman09, newman_etal01} as, in our notation,
\begin{equation}\label{eq:clustering}
    C(3) = \frac{3 \times \textrm{number of triangles in the network}}{\textrm{number of connected triples of vertices}},
\end{equation} 
where a \emph{connected triple} is a set of three distinct nodes $u$, $v$ and $w$, such that at least two of the possible edges between them exist. In a social network this measures how often an individual's ``friends'' are also friends with one another, on average, across the entire network. An alternative measure of clustering, the \emph{average clustering coefficient}, takes a different approach to (\ref{eq:clustering}) and is computed locally for each node, and then averaged across all nodes.\footnote{Overall clustering (\ref{eq:clustering}) assigns the same weight to every triangle in the graph. Average clustering gives each node the same weight. Since high-degree nodes may be adjacent to more triangles than low-degree nodes, overall and average clustering can give different values.} We focus on overall clustering (transitivity) in this paper.

There is now substantial evidence that significant topological structures (known as ``graphlets'', ``motifs'' or ``subgraphs''), on more than three nodes, can be found in real-world networks, and that they may perform precise specialized functions \cite{agasse-duval_lawford18, benson17, benson_etal16, bounova09, jeong_etal00, milo_etal02, yaveroglu_etal14}. Since the usual clustering coefficient $C(3)$ is based upon connected triples of nodes, it is natural to ask whether a similar measure can be derived for any number of nodes. A generalized clustering coefficient could potentially identify hidden higher-order clustering and enable a better understanding of the structure of real-world networks. The need to go beyond usual three-node \emph{average} clustering has been addressed by \cite{caldarelli_etal04} (cycles of length four), who find that ``grid clustering'' scales with node degree in a similar way to usual clustering, \cite{fronczak_etal02} (shortest paths of length greater than one between a node's neighbours), who confirm the absence of clustering in Barab{\'a}si-Albert preferential attachment (scale-free) models, and \cite{jiang_claramunt04} (connectivity of more than one of a node's nearest neighbours), who investigate scaling properties using data on urban street networks.

The work that is most closely related to the present paper is by Yin-Benson-Leskovec \cite{yin_etal18}, hereafter referred to as YBL, who propose new overall and average clustering coefficients based on the relative frequency of cliques of order greater than three. Our work differs essentially from YBL's \emph{overall} coefficient in the way that we define the ``relative frequency'', but also in the methods that are used for computation, and the motivation and empirical application. We draw careful comparisons between our method and YBL, for theoretical Erd{\H{o}}s-R{\'{e}}nyi random graphs and simulated small-world models, and argue that the two approaches are complementary.\footnote{We became aware of the excellent \cite{yin_etal18} after the original version of our paper had been completed and submitted to the arXiv repository. The present paper contains substantial new material to address this omission. There is related work by YBL-Gleich \cite{yin_etal17} and by YBL \cite{yin_etal19}.}

In this paper, we make the following specific contributions:
\begin{itemize}

    \item We propose a new generalized clustering coefficient $C(b)$, based upon connected groups of $b$ nodes, which nests the standard clustering coefficient $C(3)$. We develop a very fast analytic implementation for connected groups of three, four and five nodes, that we show to be up to 2,000 times faster than a na{\"i}ve nested loop algorithm, for some small dense graphs.
    
    \item We examine the theoretical properties of $C(b)$ for Erd{\H{o}}s-R{\'{e}}nyi and small-world random graphs, and lollipop graphs, and draw comparisons with YBL. We show that it will become increasingly difficult to compute $C(b)$ efficiently as $b$ becomes large, even using analytic formulae. Using dynamic data on U.S. airline networks, we also observe that $C(b)$ can be highly correlated across $b$, and with network density. When we control for lower-order clustering, we find low to moderate higher-order clustering in these networks. It is not known whether this finding holds generally for large classes of networks.
    
\end{itemize}
All of the analytic formulae that we use, and several proofs, are collected in Appendix \ref{sec:proofs}, and additional figures and tables are reported in Appendix \ref{sec:correlations}.

\section{Graph Theory and Clustering}
We briefly review some relevant tools of graph theory. Important monographs include \cite{diestel17} (mathematics), \cite{jackson08} (economics of social networks) and \cite{jungnickel08} (algorithms). A \emph{graph} is an ordered pair $G = (V, E)$ where $V$ and $E$ denote the sets of \emph{nodes} and \emph{edges} of $G$, respectively. We use $n=|V|$ and $m=|E|$ to represent the numbers of nodes and edges of $G$. A graph has an associated $n \times n$ \emph{adjacency matrix} $g$, with representative element $(g)_{ij}$ that takes value one when an edge is present between nodes $i$ and $j$, and zero otherwise. We also use $(i, j) \in E$ to denote an edge between nodes $i$ and $j$, and say that they are \emph{directly-connected}. A graph is \emph{simple and unweighted} if $(g)_{ii} = 0$ (no self-links) and $(g)_{ij} \in \{0, 1\}$ (no pair of nodes is linked by more than one edge, or by an edge with a weight that is different from one). A graph is \emph{undirected} if $(g)_{ij} = (g)_{ji}$. A \emph{walk} between nodes $i$ and $j$ is a sequence of edges $\{(i_{r}, i_{r+1})\}_{r=1,\ldots,R}$ such that $i_{1} = i$ and $i_{R+1} = j$, and a \emph{path} is a walk with distinct nodes. A graph is \emph{connected} if there is at least one path between any pair of nodes $i$ and $j$; otherwise the graph is \emph{disconnected}. A \emph{bridge} is an edge the removal of which will disconnect the graph. In this paper, we consider simple, unweighted, undirected and connected graphs.

The \emph{degree} $k_i = \sum_{j}(g)_{ij}$ is the number of nodes that are directly-connected to node $i$, and the \emph{(1-degree) neighbourhood} of node $i$ in $G$, denoted by $\Gamma_{G}(i) = \{j: (i, j) \in E\}$, is the set of all nodes that are directly-connected to $i$. The \emph{density} $d(G) = 2 m / n (n-1)$ is the number of edges in $G$ relative to the maximum possible number of edges in a graph with $n$ nodes. A graph $G' = (V', E')$ is a \emph{subgraph} of $G$ if $V' \subseteq V$ and $E' \subseteq E$ where $(i, j) \in E'$ implies that $i, j \in V'$. A \emph{tree} is a connected graph with no cycles. A \emph{spanning tree} on a connected $G$ is a connected subgraph with nodes $V$ and the minimum possible number of edges $m=n-1$. A \emph{complete} graph on $n$ nodes, $K_{n}$, has all possible edges, and a complete subgraph on $b$ nodes is called a $b$\emph{-clique}. A \emph{maximal clique} is a clique that cannot be made larger by the addition of another node in $G$ with its associated edges, while preserving the complete-connectivity of the clique. A \emph{maximum clique} is a (maximal) clique of the largest possible size in $G$, and the \emph{clique number} $w(G)$ of the graph $G$ is the number of nodes in a maximum clique in $G$.

Let $G(n, p)$ be an Erd{\H{o}}s-R{\'{e}}nyi random graph with nodes $V = \{1,\ldots,n\}$ and edges that arise independently with a constant edge-formation probability $p$, giving a statistically homogeneous network that has, on average, $(n-1) \, p$ edges for a given node, and $\binom{n}{2} \, p$ randomly-distributed edges in total. We also use the lollipop graph $L(b, n-b)$, with $n$ nodes and $\frac{1}{2}b(b-3)+n$ edges, and $3 \leq b \leq n$. The lollipop can be thought of as a $b$-clique $K_b$ that is attached by a bridge to a path graph on $n-b$ nodes. Note that $L(n, 0)$ is the complete graph $K_n$.\footnote{The lollipop was first introduced by \cite[Example 2]{lawler86} in the one parameter case $b = n/2$, and was generalized to two parameters by \cite{brightwell_winkler90}. It has applications in the fields of combinatorics (Ramsey theory) \cite{fox08} and linear algebra (spectral theory) \cite{boulet_jouve08, haemers_etal08}.} Using the notation of \cite{agasse-duval_lawford18}, we refer to particular topological subgraphs by $M_{a}^{(b)}$, where $b$ is the number of nodes in the subgraph, and $a$ is the decimal representation of the smallest binary number derived from a row-by-row reading of the upper triangles of each adjacency matrix $g$ from the set of all topologically-identical subgraphs on the same $b$ nodes (also see \cite{lawford20}).

\subsection{Analytic formulae for a generalized clustering coefficient}

    The clustering coefficient $C(3)$ is bounded by $0 \leq C(3) \leq 1$, attaining the minimum when there are no triangles in the graph, and taking the maximum value for a complete graph $K_n$. Since each triangle contains three triples of nodes, a factor of three appears in the numerator of (\ref{eq:clustering}). A na{\"i}ve algorithm that is based on nested loops, and considers every distinct triple of nodes in $G$, will run in $O(n^{3})$ time. However, it is easy to write down an analytic version of $C(3)$, using the nested subgraph enumeration formulae in \cite[equations (1) and (2)]{agasse-duval_lawford18}:

    \begin{equation}\label{eq:c_3_analytic}
        C(3) = \frac{3 \, |M_{7}^{(3)}|}{|M_{3}^{(3)}|} = \frac{\tr(g^{3})}{\sum_{i}k_{i}(k_{i} - 1)},
    \end{equation}
    and where $C(3)$ makes explicit the definition of clustering in terms of triples $M_{3}^{(3)}$ and triangles $M_{7}^{(3)}$.
    
    If we instead interpret (\ref{eq:c_3_analytic}) as the average probability that any three connected nodes in a graph are also completely-connected, then a natural generalization follows to any number $b$ of nodes, such that $3 \leq b \leq n$. In this paper, we define the \emph{generalized clustering coefficient} as follows:
    
    \begin{equation}\label{eq:c_b}
        C(b) = \frac{a(b) \times \textrm{number of $b$-cliques in $G$}}{\textrm{number of $b$-spanning trees in $G$}},
    \end{equation}
    where \emph{Cayley's formula} $a(b) = b^{b-2}$ gives the number of spanning trees in $K_{b}$, and ensures that $0 \leq C(b) \leq 1$. Clearly, $C(b)$ nests $C(3)$, and equals zero if and only if there are no $b$-cliques in the graph. It is natural that the generalized clustering coefficient should attain its maximum value for a complete graph, in the same way as $C(3)$, and we show this in:
    \begin{proposition}\label{thm:case_c(b)=1}
        Let $G$ be a connected graph with at least $b$ nodes $(b\geq 3)$. Then $C(b) =1$ if and only if G is complete.
    \end{proposition}
    \noindent See Appendix \ref{sec:proofs} for a proof of Proposition \ref{thm:case_c(b)=1}.
    
    A na{\"i}ve algorithm for (\ref{eq:c_b}), based on nested loops, will run in $O(n^{b})$ time. For example, the denominator of (\ref{eq:c_b}) can be calculated by considering every distinct set of $b$ nodes in $G$, and counting the number of spanning trees on each subgraph. This will be excessively slow. If we instead think of $C(b)$ as a measure of the prevalence of $b$-cliques relative to all connected groups of $b$ nodes, then it is clear that we can use analytic subgraph enumeration for counting the cliques and the spanning trees for the special cases $C(4)$ and $C(5)$, in the same way as for (\ref{eq:c_3_analytic}):
    
    \begin{equation}\label{eq:c_4_analytic}
        C(4) = \frac{16 \, |M_{63}^{(4)}|}{|M_{11}^{(4)}| + |M_{13}^{(4)}|} = \frac{4 \, \sum_{i}\tr(g_{-i}^{3})}{\sum_{i}k_{i}(k_{i} - 1)(k_{i} - 2) + 6 \, \sum_{(i, j) \in E}(k_{i} - 1)(k_{j} - 1) - 3 \, \tr(g^{3})}.
    \end{equation}

    \begin{equation}\label{eq:c_5_analytic}
        \begin{split}
            C(5) &= \frac{125 \, |M_{1023}^{(5)}|}{|M_{75}^{(5)}| + |M_{77}^{(5)}| + |M_{86}^{(5)}|} \\
                 &= \frac{25 \, \sum_{i} \sum_{j \in \Gamma_{G}(i)}\tr(((g_{-i})_{-j})^{3})}{\splitfrac{\sum_{i}k_{i}(k_{i} - 1)\{(k_{i} - 2)(k_{i} - 15) - 24\} + 12 \, \sum_{(i,j) \in E}(k_{i} - 1)(k_{i} + k_{j} - 8)(k_{j} - 1)}{- 48 \, \sum_{i}(g^{3})_{ii}(k_{i} - 2) + 12 \, \sum_{i \neq j}(g^{4})_{ij} - 12 \, \tr(g^{3})}}.
        \end{split}
    \end{equation}
    The numerator terms $|M_{63}^{(4)}|$ and $|M_{1023}^{(5)}|$ are the number of 4-cliques and 5-cliques respectively. The denominator terms are the counts of the 4-star ($|M_{11}^{(4)}|$), the 4-path ($|M_{13}^{(4)}|$), the 5-star ($|M_{75}^{(5)}|$), the 5-arrow ($|M_{77}^{(5)}|$), and the 5-path ($|M_{86}^{(5)}|$), which are illustrated in Figures \ref{fig:connected_four_nodes} and \ref{fig:5_node_subgraph_notation}. Since there are sixteen possible spanning trees on any given four nodes in the graph, all of which will occur in $K_{4}$, the factor $a(4)$ equals 16. Similarly, counting the distinct 5-spanning trees in $K_{5}$ gives $a(5)$ equal to 125. We do not recommend using the right-hand-sides of (\ref{eq:c_4_analytic}) and (\ref{eq:c_5_analytic}) for computation. We report the runtime performance of the analytic formulae in Appendix \ref{sec:runtime} for small graphs.
    
    \begin{figure}\centering
        	\includegraphics[scale=0.15]{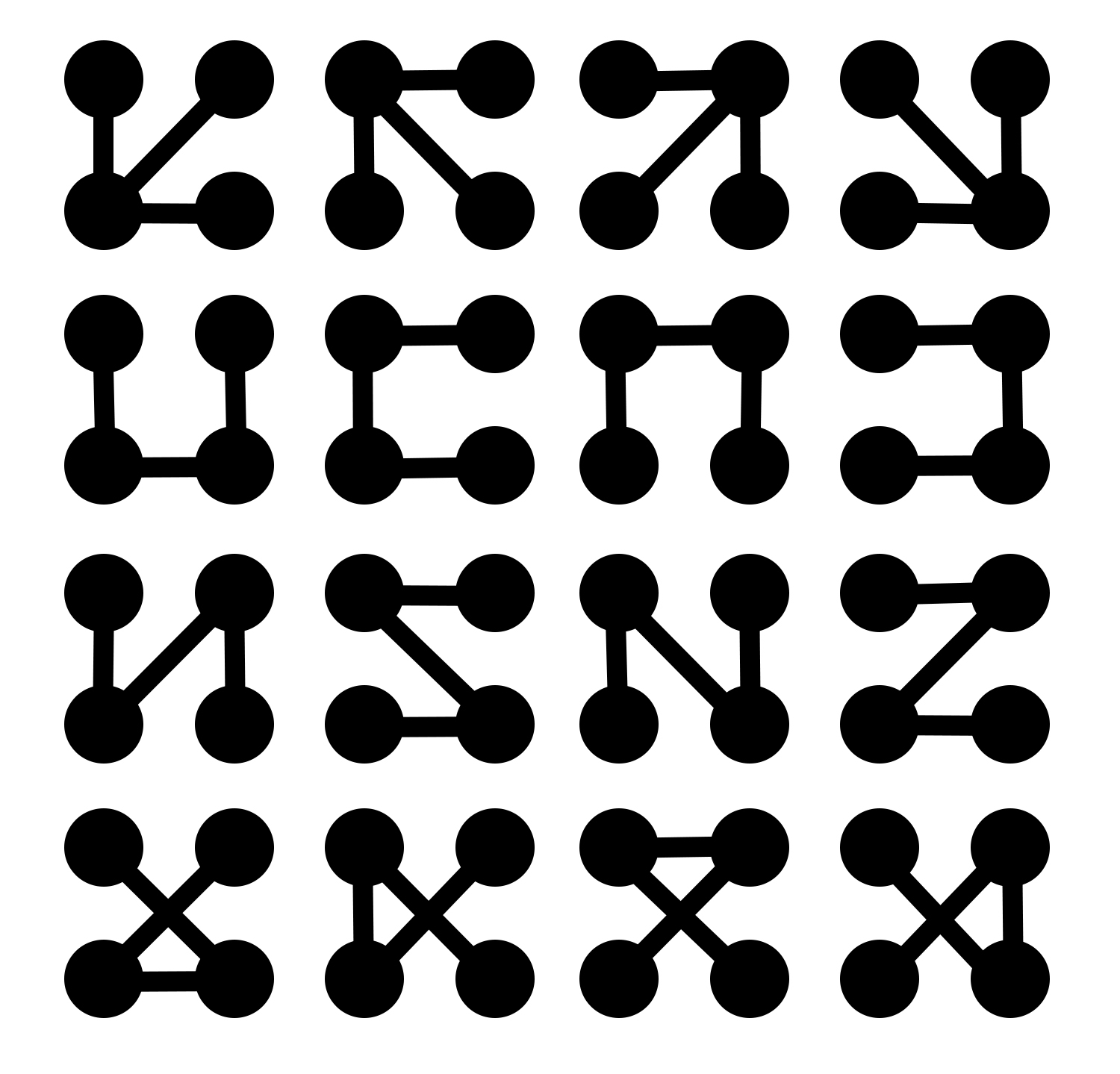}
        	\caption{The sixteen spanning trees on four labelled nodes: four 4-stars $M_{11}^{(4)}$ and twelve 4-paths $M_{13}^{(4)}$. This illustrates Cayley's formula $a(b) = b^{b-2}$, which appears in the numerator of the generalized clustering coefficient, for $b=4$.}
        	\label{fig:connected_four_nodes}
    \end{figure}

    \begin{figure}\centering
    	\begin{subfigure}{0.20\textwidth}
    		\centering
    		\includegraphics[width=.5\linewidth]{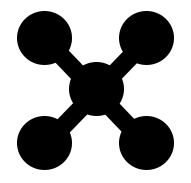}
    		\caption{5-star $M_{75}^{(5)}$.}
    		\label{fig:m_75_5}
    	\end{subfigure}
    	\begin{subfigure}{0.20\textwidth}
    		\centering
    		\includegraphics[width=.5\linewidth]{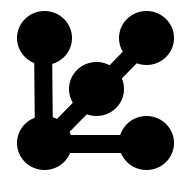}
    		\caption{5-arrow $M_{77}^{(5)}$.}
    		\label{fig:m_77_5}
    	\end{subfigure}
    		\begin{subfigure}{0.20\textwidth}
    		\centering
    		\includegraphics[width=.5\linewidth]{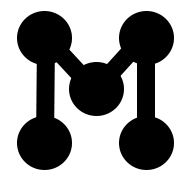}
    		\caption{5-path $M_{86}^{(5)}$.}
    		\label{fig:m_86_5}
    	\end{subfigure}
    	\caption{The three spanning trees on five unlabelled nodes: the 5-star $M_{75}^{(5)}$, the 5-arrow $M_{77}^{(5)}$ and the 5-path $M_{86}^{(5)}$. The total count of these subgraphs appears in the denominator of the generalized clustering coefficient $C(b)$, for $b=5$.}
    	\label{fig:5_node_subgraph_notation}
    \end{figure}

\begin{figure}\centering
                  \begin{subfigure}[b]{0.15\linewidth}
                    \centering
                   \includegraphics[scale=0.14]{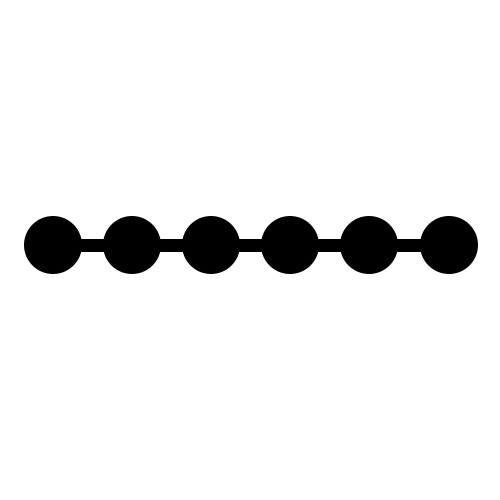}
                  \end{subfigure}
                  \begin{subfigure}[b]{0.15\linewidth}
                    \centering
                   \includegraphics[scale=0.14]{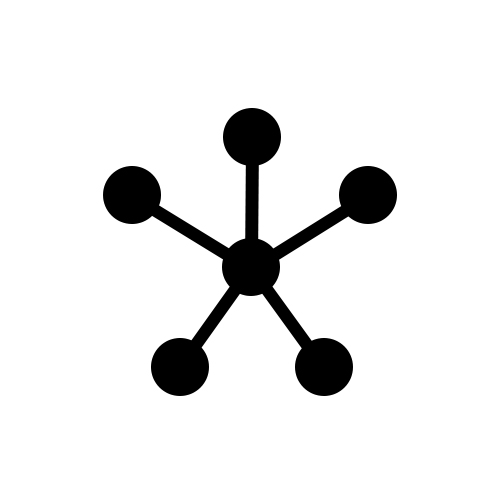}
                  \end{subfigure}
                  \begin{subfigure}[b]{0.15\linewidth}
                    \centering
                   \includegraphics[scale=0.14]{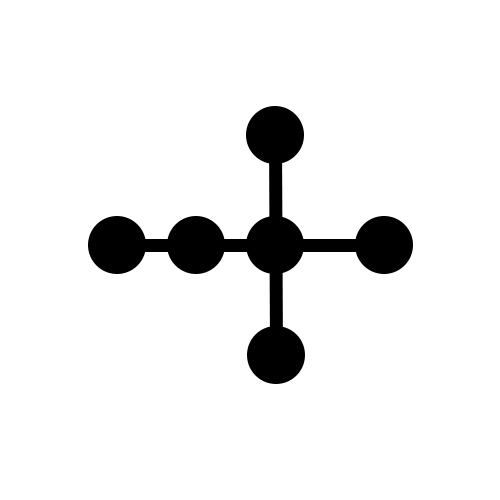}
                  \end{subfigure}
                  \begin{subfigure}[b]{0.15\linewidth}
                    \centering
                   \includegraphics[scale=0.14]{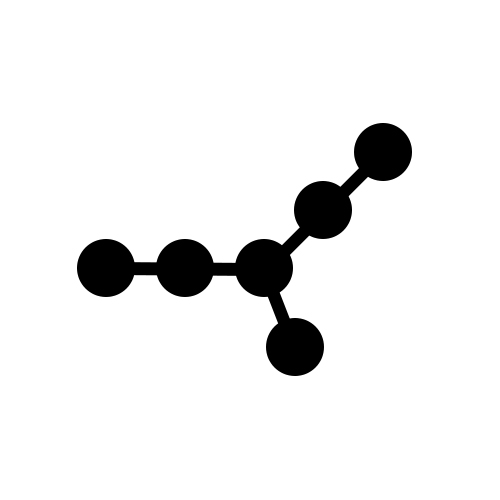}
                  \end{subfigure}
                  \begin{subfigure}[b]{0.15\linewidth}
                    \centering
                   \includegraphics[scale=0.14]{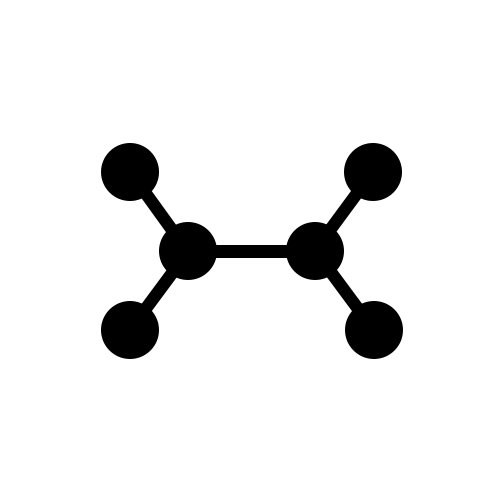}
                  \end{subfigure}
                  \begin{subfigure}[b]{0.15\linewidth}
                    \centering
                   \includegraphics[scale=0.14]{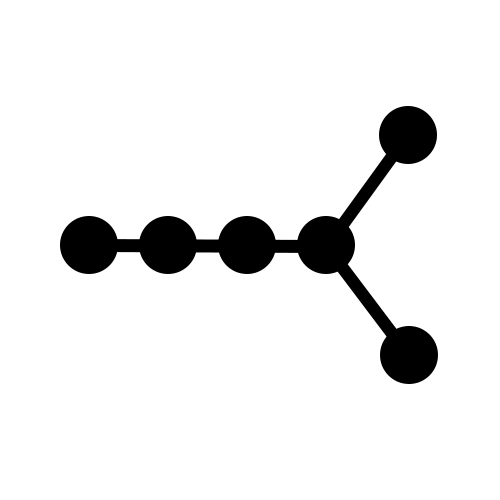}
                  \end{subfigure}
                  
                  \caption{The six non-isomorphic spanning trees on six unlabelled nodes. The total count of these subgraphs appears in the denominator of the generalized clustering coefficient $C(b)$, for $b=6$. The number of spanning trees that must be counted in order to compute $C(b)$ increases very rapidly with the number of nodes $b$ in the subgraph.}\label{fig:non_isomorphic_trees_six_nodes}
            \end{figure}
            
However, while this approach seems promising, it will rapidly become hard to derive analytic formulae for larger values of $b$, because the number of denominator terms will explode. Essentially, we would need to find a formula for \emph{every} non-isomorphic tree on $b$ nodes. For example, $C(6)$ would require evaluation of six denominator terms (Figure \ref{fig:non_isomorphic_trees_six_nodes}). Numerical values for the number of trees on $n$ unlabelled nodes are given as series A000055 in the Online Encyclopedia of Integer Sequences (\ \url{http://oeis.org/A000055}\ ). For example, $C(7)$ has 11 denominator terms, $C(8)$ has 23 denominator terms, and $C(36)$ has more than 6.2 $\times$ $10^{12}$ terms! This creates an intrinsic bound on the general applicability of analytic formulae for $C(b)$: we can reasonably expect to use them for $C(3)$, $C(4)$, $C(5)$ and perhaps $C(6)$ and $C(7)$, but not beyond. There has been considerable research on efficient numerical algorithms for generating all possible spanning trees of a simple undirected connected graph; see \cite{chakraborty_etal19} for a review and comparison of different methods. It is possible that numerical methods could be used to extend $C(b)$ to higher values of $b$, although computation of $C(b)$ requires consideration of \emph{every} possible set of $b$ connected nodes in the graph, and it is well known that the number of spanning trees of a graph increases exponentially in the number of nodes.\footnote{Note that the analytical algorithm is not a property of the generalized clustering coefficient itself, but a means to compute it efficiently for small $b$, and for small to moderate $n$. As in other areas, when exact analytic methods become intractable, it often becomes necessary to use numerical techniques and approximations instead, such as approximate sampling algorithms and asymptotic results.}

\subsection{The generalized clustering coefficient $C_{b-1}$ of Yin-Benson-Leskovec}
In recent work, Yin-Benson-Leskovec \cite{yin_etal18} develop an overall generalized clustering coefficient, based on clique expansion, that is closely related to $C(b)$.\footnote{``The novelty of our interpretation of the clustering coefficient is considering it as a form of clique expansion rather than as the closure of a length-2 path, which is key to our generalizations in the next section.'' (Page 052306-2 in \cite{yin_etal18})} Their coefficient (Equation 4 in \cite{yin_etal18}) is
\begin{equation*}
    C_{\ell} = \frac{(\ell^{2}+\ell)\, |K_{\ell+1}|}{|W_{\ell}|}; \quad \ell \geq 2,
\end{equation*}
where $K_{\ell+1}$ is the set of $(\ell+1)$-cliques and $W_{\ell}$ is the set of $\ell$ ``wedges'' (they define a wedge as an $\ell$-clique with one additional node that is adjacent to any node in the clique). Despite the apparently different formulation, we can write YBL's coefficient in the notation of our paper, with $\ell = b-1$, as
\begin{equation*}
    C_{b-1} = \frac{(b^{2}-b) \, |K_{b}|}{|L(b-1,1)|}; \quad b\geq 4,
\end{equation*}
where $L(\cdot,\cdot)$ is the lollipop graph. Note that the lollipop $L(2,1)$ is not typically defined, and so we need $b-1 \geq 3$. YBL nest the usual clustering coefficient in their generalized framework by implicitly defining $L(2,1)$ as a 3-path with directed edges (double-counting the undirected 3-path), which gives the coefficient 6 in the numerator rather than the usual 3. This is just a counting issue, and so we assume in the rest of the paper that $C(3)=C_{2}$.

We now compare the higher-order clustering coefficients on four and five nodes:
\begin{equation*}
    C(4) = \frac{16 \, |K_{4}|}{\textrm{number of $4$-spanning trees in $G$}} = \frac{16 \, |M_{63}^{(4)}|}{|M_{11}^{(4)}|+|M_{13}^{(4)}|}; \quad C_{3} = \frac{12 \, |K_{4}|}{|L(3,1)|} = \frac{12 \, |M_{63}^{(4)}|}{|M_{15}^{(4)}|},
\end{equation*}
where $M_{15}^{(4)}$ is the tadpole subgraph; and
\begin{equation*}
    C(5) = \frac{125 \, |K_{5}|}{\textrm{number of $5$-spanning trees in $G$}} = \frac{125 \, |M_{1023}^{(5)}|}{|M_{75}^{(5)}|+|M_{77}^{(5)}|+|M_{86}^{(5)}|}; \quad C_{4} = \frac{20 \, |K_{5}|}{|L(4,1)|} = \frac{20 \, |M_{1023}^{(5)}|}{|M_{127}^{(5)}|},
\end{equation*}
where $M_{127}^{(5)}$ is the kite subgraph (see \cite{lawford20} for the count formula). Hence YBL's coefficient $C_{b-1}$ fits naturally into the analytic framework of our paper. Likewise, our coefficient $C(b)$ can potentially use similar computational techniques to those in YBL. Intuitively, both methods will face computational difficulties as $b$ becomes large, and in fact YBL do not go beyond $b=5$. In practice, this is unlikely to be a serious issue, as shown by the empirical results of YBL, and by Section \ref{sec:empirical} in this paper. Small values of $b$ appear to be sufficient to capture much of the higher-order clustering that is present in some real-world networks.

The essential difference between $C(b)$ and $C_{b-1}$ is in the definition of the ``relative frequency''. We compute the frequency of $b$-cliques relative to the number of minimally connected subgraphs on $b$ nodes. On the other hand, YBL compute the frequency of $b$-cliques relative to the number of lollipops $L(b-1,1)$ i.e. a $b$-clique where all but one of the edges adjacent to one node have been removed. This reflects the two possible interpretations of the 3-path in the denominator of $C(3)$, either as a spanning tree on three nodes (our paper) or as a 2-clique with an additional adjacent node (YBL), and the two natural generalizations of $C(3)$ to higher-order clustering.

\subsection{Analysis of generalized clustering $C(b)$ and $C_{b-1}$ for the $G(n, p)$ model}
In the special case of the Erd{\H{o}}s-R{\'{e}}nyi random graph $G=G(n, p)$, it follows from (\ref{eq:c_b}) that the expectation of $C(b)$ is given by $\mathbb{E}_{G}[C(b)] = p^{(b-1)(b-2)/2}$ as $n$ becomes large, since there are $\binom{n}{b} \, p^{\binom{b}{2}}$ $b$-cliques and $b^{b-2} \, \binom{n}{b} \, p^{b-1}$ $b$-spanning trees in $G(n, p)$, on average. Also see Figure \ref{fig:G_n_p_clustering}. We implicitly assume that both $C(b)$ and $C_{b-1}$ are well-defined on $G$. Numerical values of $\binom{b-1}{2}$ are given in  A161680 of the Online Encyclopedia of Integer Sequences (\ \url{http://oeis.org/A161680}\ ).
\begin{figure}\centering
    \includegraphics[scale=0.35]{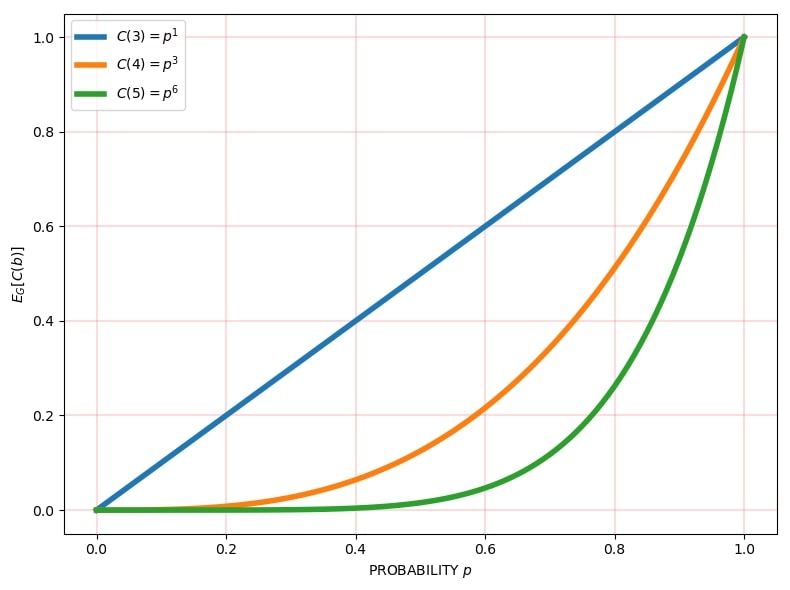}
    \caption{The theoretical expectation of the clustering coefficient $C(b)$ for the Erd{\H{o}}s-R{\'{e}}nyi random graph $G(n, p)$ is $\mathbb{E}_{G}[C(b)]=p^{(b-1)(b-2)/2}$ as $n$ becomes large, with edge-formation probability $0 \leq p \leq 1$. We observe that clustering is monotonically increasing in probability $p$, as the network moves from a set of disconnected nodes to a complete graph. For a given $p$, the expected level of clustering is decreasing in $b$.}\label{fig:G_n_p_clustering}
\end{figure}
We can also show that $\mathbb{E}_{G}[C_{b-1}]=p^{b-2}$ as $n$ becomes large, since there are $\binom{n}{b} \, p^{\binom{b}{2}}$ $b$-cliques and $(b^{2}-b) \, \binom{n}{b} \, p^{\binom{b-1}{2}+1}$ lollipops in $G(n, p)$, on average (see also Proposition 2(1) in \cite{yin_etal18} for this result). It follows immediately that $\mathbb{E}_{G}[C(b)] \gtreqqless \mathbb{E}_{G}[C_{b-1}]$ as $(b-2)(b-3) \lesseqqgtr 0$, with $p \neq 0,1$. Directly, $\mathbb{E}_{G}[C(b)] \leq \mathbb{E}_{G}[C_{b-1}]$, with equality when $p=0,1$ (for all $b$) or when $b=3$ (for all $p$). Intuitively, there will be more $b$-spanning trees than there are lollipops $L(b-1,1)$ in $G$. For example, there is one 4-star and two 4-paths in the tadpole; and there can also be 4-spanning trees that are not in a tadpole. In Figure \ref{fig:clustering_difference} we examine the expected difference between $C(b)$ and $C_{b-1}$ for $G=G(n, p)$, where $\mathbb{E}_{G}[C_{b-1}] - \mathbb{E}_{G}[C(b)] = p^{b-2} \, (1-p^{(b-2)(b-3)/2}) \geq 0$, with equality at $p=0,1$. We see that the difference can be substantial, and that it increases in the edge-formation probability $p$ (as the graph $G$ becomes more dense) up to a certain point, and then falls to zero; but that it can increase or decrease in $b$, the order of the clustering, for a given $p$.

\begin{figure}\centering
    \includegraphics[scale=0.35]{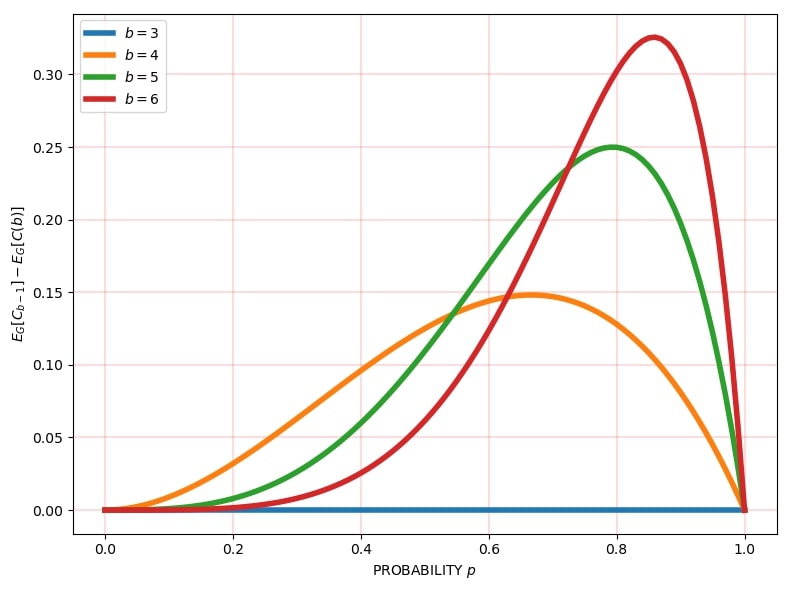}
    \caption{The theoretical difference in expectation between the clustering coefficient $C_{b-1}$ of \cite{yin_etal18} and our coefficient $C(b)$ for the Erd{\H{o}}s-R{\'{e}}nyi random graph $G(n, p)$ is $\mathbb{E}_{G}[C_{b-1}] - \mathbb{E}_{G}[C(b)]=p^{b-2}(1-p^{(b-2)(b-3)/2})$ as $n$ becomes large, with edge-formation probability $0 \leq p \leq 1$. We note that $C(3)=C_{2}$. We observe that $C(b)$ and $C_{b-1}$ are numerically identical when $p=0$ (a set of disconnected nodes) and $p=1$ (a complete graph). For a given $b$, the expected difference is positive and increases in $p$ up to a certain point, and then falls to zero. It is easy to show that the maximum is attained at $p=(2/(b-1))^{2/(b-2)(b-3)}$, which gives $p=2/3$ (for $b=4$), $p=2^{-1/3}\approx 0.7937$ (for $b=5$) and $p=(2/5)^{1/6}\approx 0.8584$ (for $b=6$). It is clear that $C_{b-1} > C(b)$ in expectation for all $p \neq 0,1$. The numerical difference between the two statistics can be quite substantial, given that $0 \leq C_{b-1} \leq 1$ and $0 \leq C(b) \leq 1$. The maximum difference is $4/27\approx 0.1481$ (for $b=4$), $1/4$ (for $b=5$) and $(108/3125)^{1/3}\approx 0.3257$ (for $b=6$).}\label{fig:clustering_difference}
\end{figure}

\subsection{Invariance of $C_{b-1}$ on graphs with vanishing density}
The YBL coefficient $C_{b-1} \, (b \geq 4)$ has a peculiar invariance property for a particular class of graphs, that is neither a feature of the usual clustering coefficient $C(3)=C_{2}$ nor of our generalized clustering $C(b)$. Note that $C(b)$ increases as the number of $b$-cliques increases or the number of $b$-spanning trees falls, while $C_{b-1}$ increases as the number of $b$-cliques increases or the number of lollipops $L(b-1,1)$ falls. It follows that $C_{b-1}$ will give the \emph{same} value for a lollipop graph $G=L(b,1)$ as for \emph{any} graph $G_{2}$ on more than $b+1$ nodes that does not contain any additional lollipops $L(b-1,1)$ beyond those contained in $G$.\footnote{If $G_{2}$ contains any positive contribution to higher-order clustering (in the sense of more $b$-cliques) then both $C(b)$ and $C_{b-1}$ will detect it.} We illustrate with two specific examples. First, take the lollipop $L(3,n-3)$ with $n \geq 4$. It is easy to show that $C(3)=C_{2}=3/(n+1)$, which is decreasing in $n$. This makes intuitive sense: as the lollipop becomes progressively less dense, there is less clustering, and in the limit $n \to \infty$ the lollipop resembles a path graph. Second, take the lollipop $L(4, n-4)$ that has $C(4)=16/(n+22)$ for $n \geq 6$: our fourth order clustering decreases in $n$. The usual clustering $C(3)=12/(n+10)$ (for $n \geq 5$) also falls in $n$. Again, this is intuitively correct: the density of $L(4,n-4)$, for $n \geq 5$, is $d(G)=2(n+2)/(n(n-1)) \to 0$ as $n \to \infty$. However, $C_{3}=12 \, |M_{63}^{(4)}|/|M_{15}^{(4)}|=0.8$ for all $n \geq 5$ i.e. it is invariant as $n$ increases.\footnote{We observe qualitatively the same result for the lollipop $L(5,n-5)$ with $C(3)=30/(n+28)$ for $n \geq 6$ and $C(5)=125/(n+203)$ for $n \geq 8$ but $C_{4}=5/6\approx 0.8333$ for $n \geq 6$ despite a vanishing density $d(G)=2(n+5)/(n(n-1)) \to 0$ as $n \to \infty$.}

This is interesting for two reasons. First, $C_{3}$ will not change even as the graph becomes infinitely sparse in the limit. Second, $C_{3}$ does not behave in the same way as the usual clustering coefficient $C_{2}$ in this respect. The problem is not specific to the lollipop, and $C_{b-1}$ will display this behaviour whenever a lollipop $L(b,1)$ is extended in a non-trivial way so that no additional $L(b-1,1)$ structure is added to the graph. Consider $G=L(5,1)$ with $n=6$ nodes and $m=11$ edges, so that $C_{4} \approx 0.8333$. Then consider a connected graph $G_{2}$ with an arbitrarily large number of nodes, that includes one copy of $G$ as a subgraph, and that has additional non-trivial topological structure outside of the subgraph $G$, that can include 17 of the possible 21 non-isomorphic subgraphs on five nodes \cite{lawford20}, but with no additional $L(4,1)$ structure. Nevertheless, $C_{4}$ will take the same value on $G_{2}$ as on $G$.

\subsection{Analysis of generalized clustering $C(b)$ for the small-world model}
We now investigate the properties of $C(b)$ for a simulated small-world model that can interpolate between regular and random behaviour. Many real-world networks exhibit small-world behaviour \cite{albert_barabasi02, barrat_weigt00, jackson_rogers05, marvel_etal13, newman00, watts_strogatz98, zaidi12}. Following the original one-parameter model of \cite{watts_strogatz98}, we start with a regular ring lattice on $n$ nodes, where each node is adjacent to its $k$ nearest neighbours in both the clockwise and anti-clockwise directions, giving a total $nk$ edges. To have a sparse but connected network, we assume that $n \gg 2k \gg \ln(n) \gg 1$.\footnote{If $n=2k+1$ then the graph is complete, and $C(b)=1$ from Proposition \ref{thm:case_c(b)=1}, for $3 \leq b \leq n$. If $n>2k+1$ then, trivially, we have that $C(b)=0$ for $b>k+1$ on the regular ring lattice with no rewiring.} We choose an edge that connects a node $u$ to its nearest neighbour $v$ in a clockwise direction, and rewire this edge, uniformly with probability $p$, to an edge $(u,w)$, avoiding self-loops and duplicated edges. With probability $1-p$ the original edge is left in place. We continue in a clockwise direction around the ring until all nodes have been considered once. We then look at edges that connect each node to its second-nearest neighbour, and so on, continuing around the ring $k$ times, until each edge in the original lattice has been examined once. The edge rewiring creates more randomness in the network: when $p=0$, the ring lattice is unchanged and completely regular, and when $p=1$ all edges are rewired randomly. It is well-known \cite{watts_strogatz98} that intermediate values of $p$ give graphs with the small-world property, characterized by low average path length (as in a random graph) but high clustering (as in a regular graph). This is due to the presence of a small number of ``short-cut'' edges that connect nodes that would otherwise be far apart in a regular graph.

In Figure \ref{fig:small_world_n50_k7_M200}, we report the expected clustering coefficient $\mathbb{E}_{G}[C(b)]$ from 250 replications of a small-world graph $G$ on $n=50$ nodes, with $k=7$. The small-world model is connected by construction, and so $C(b)$ will be well defined. We note that clustering falls in $b$ (given $p$) but that it is not monotonic as $p$ increases (given $b$): at some point, introducing more randomness actually increases clustering.\footnote{The observation that expected clustering falls in $b$ for a given $p$ is qualitatively the same as Figure 3 in \cite{yin_etal18}, who also consider this small-world model, but for their \emph{average} clustering coefficient, and with (apparently one replication of) a small-world model on $n=20,000$ nodes, and $k=5$. Figure 3 in \cite{yin_etal18} also suggests that their average clustering decreases monotonically in $p$ (for a given $b$) for this large $n$. See Figure \ref{fig:small_world_n50_k7_M200_ybl} for simulation results on the expected overall clustering coefficient $\mathbb{E}_{G}[C_{b-1}]$ of Yin-Benson-Leskovec, and Figure \ref{fig:clustering_difference_small_world} for simulation results on the difference in expected clustering $\mathbb{E}_{G}[C_{b-1}] - \mathbb{E}_{G}[C(b)]$, for the small-world model of \cite{watts_strogatz98}.} Finding a closed-form formula for the expected value of $C(b)$ in a finite small-world model is an open problem. Some partial results are available \cite{albert_barabasi02, barrat_weigt00, watts_strogatz98}. When $p=0$, we have $C(3)=3/2 \times (k-1)/(2k-1)$ for all $n$ large relative to $k$. There are $n\,\binom{k}{2}$ triangles as $n$ is large relative to $k$, and $n\,k\,(2k-1)$ connected triples. In the example of Figure \ref{fig:small_world_n50_k7_M200}, $C(3) = 9/13 \approx 0.6923$ for $p=0$. Asymptotically ($n \to \infty$), it can be shown that $C(3) \approx 3/2 \times (k-1)/(2k-1) \times (1-p)^{3}$ when $p \neq 0,1$, and $C(3) \sim 2k/n$ when $p=1$. So, $C(3)$ is asymptotically monotonically decreasing in $p$ from $3/2 \times (k-1)/(2k-1)$ to zero.

\begin{figure}\centering
    \includegraphics[scale=0.4]{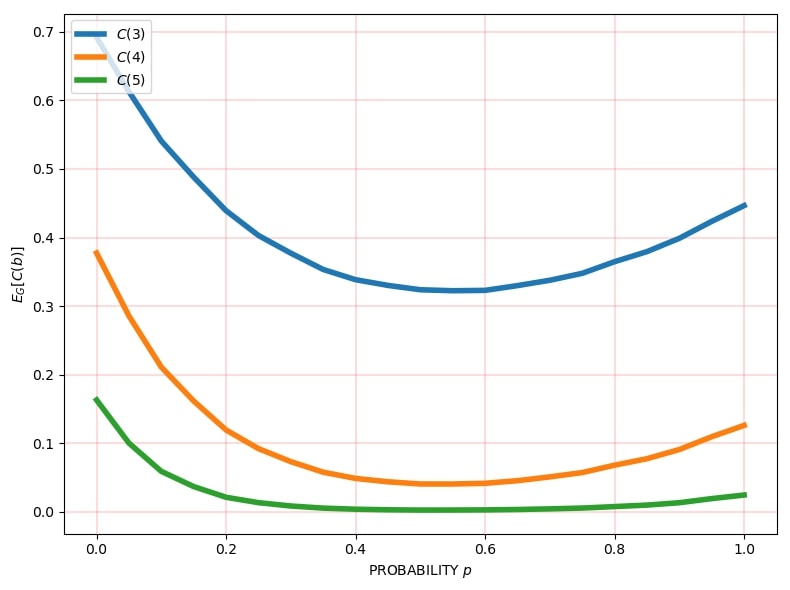}
    \caption{The simulated expected clustering coefficient $\mathbb{E}_{G}[C(b)]$ from 250 replications of a small-world graph with $n=50$ nodes, each of which has degree $2k=14$, and edge-rewiring probability $0 \leq p \leq 1$, which parameterizes the network from a regular graph to a random one. As in \cite{watts_strogatz98}, we begin with a regular ring lattice, where each node is connected to its $k$ nearest neighbours in both the clockwise and anti-clockwise directions. Moving around the lattice in a clockwise direction, each edge between a node $u$ and its nearest neighbour $v$ is randomly and uniformly rewired with probability $p$ to another edge $(u,w)$, where self-loops and repeated edges are not allowed. We then continue with the second nearest neighbour, and so on. We observe that clustering falls in $b$ but that it is not monotonic as $p$ increases. As $p$ increases from $p=0$, the randomness that is added to the graph ``breaks'' the regular clusters and reduces $C(b)$. At some point, though, introducing more randomness actually increases clustering. Finding a closed-form formula for the expectation of $C(b)$ in a small-world model with $n<\infty$ remains an open problem, except for $b=3$ and $p=0$ whereupon $C(3)=3/2 \times (k-1)/(2k-1)$ (see \cite{albert_barabasi02}), which equals $9/13 \approx 0.6923$ in this example.}\label{fig:small_world_n50_k7_M200}
\end{figure}

\section{Empirical Results on Air Transport Networks}\label{sec:empirical}
To illustrate the behaviour of $C(b)$, we construct quarterly networks for eight airline carriers over the period 1999Q1 to 2013Q4, using publicly-available data from the U.S. Department of Transportation's DB1B Airline Origin and Destination survey and T-100 Domestic Segment (All Carriers) census.\footnote{The carriers are American Airlines (AA), Alaska Airlines (AS), Delta Air Lines (DL), AirTran Airways (FL), Spirit Airlines (NK), United Airlines (UA), US Airways (US), and Southwest Airlines (WN).} The DB1B is a 10\% random sample of quarterly ticket-level itineraries, collected from reporting carriers. The T-100 is a monthly 100\% census on domestic nonstop flight segments, including number of enplaned passengers and available capacity. Both datasets have been widely used in empirical work in economics e.g. \cite{aguirregabiria_ho12, ciliberto_tamer09, dai_etal14}. We do not observe the actual date of flight or purchase, ticket restrictions, or the buyer's characteristics.

We merge the DB1B and T-100, retaining all scheduled nonstop round-trip tickets, for domestic carriers, between airports in the continental U.S. We do not keep tickets that were sold under a codesharing agreement, that have unusually high or low fares, or that are considered unreliable by the data provider. Some carriers (e.g. JetBlue Airways and Southwest Airlines) report large numbers of business and first-class tickets. We only use coach class tickets, unless more than 75\% of a carrier's tickets are listed as business or first-class, in which case we keep all tickets for that carrier.\footnote{The data treatment is quite standard in the empirical air transport literature e.g. \cite{ciliberto_williams10} motivate our filtering of tickets by fare class.} Individual tickets are then aggregated to non-directional route-carrier observations. We omit route-carriers with an especially low number of passengers, that do not have a constant number of passengers on each segment, or that are not present over the full sample period. In building the route networks, a node is an airport that was served as a route origin or destination, and an edge is present if some passengers travelled on a direct route between two nodes, for a given carrier-quarter. Our eight empirical networks are connected in every quarter of the sample. Further details of the data treatment are available from the authors.

\subsection{Descriptive statistics and small-world characteristics of airline networks}
Table \ref{tab:network_descriptives} reports descriptive statistics on the eight carrier networks in our dataset for 2013Q4. The networks are generally small (nodes and edges) and sparse (density 6--23\%). We compute the average path length and clustering coefficient $C(b)$ for $b=3,4,5$ for each real-world network, and compare these to values from simulated Erd{\H{o}}s-R{\'{e}}nyi random graphs $G(n, p)$, with $n$ equal to the number of nodes in each observed network, and the edge-formation probability $p$ set equal to its density.\footnote{The clustering coefficients reported in Table \ref{tab:network_descriptives} for $G(n,p)$ are based upon simulations and are not the theoretical asymptotic values. See Table \ref{tab:network_descriptives_ybl} for descriptive statistics on the Yin-Benson-Leskovec statistic $C_{b-1}$.} As in \cite{watts_strogatz98}, there is some evidence that most U.S. airline networks are small-world, with average path lengths that are close to those from a random graph, but higher levels of clustering.

\begin{table}
\scriptsize
\begin{center}
\begin{tabular}{lrrcccccccccc}
Carrier & Nodes & Edges & Density & apl & $\mathrm{apl}_{\mathrm{rand}}^{\mathrm{conn}}$ & $C(3)$ & $C(3)_\mathrm{rand}$ & $C(4)$ & $C(4)_\mathrm{rand}$ & $C(5)$ & $C(5)_\mathrm{rand}$ & Connected \% \\ \hline
 & & & & & & & & & & & & \\ 
AA & 71 & 153 & 0.06 & 1.94 & 3.01 & 0.120 & 0.061 & 0.018 & 0.000 & 0.002 & 0.000 & 44.6 \\ 
AS & 34 & 49 & 0.09 & 2.00 & 3.12 & 0.037 & 0.085 & 0.000 & 0.001 & 0.000 & 0.000 & 18.2 \\ 
DL & 85 & 221 & 0.06 & 1.98 & 2.84 & 0.146 & 0.061 & 0.021 & 0.000 & 0.002 & 0.000 & 65.4 \\ 
FL & 38 & 78 & 0.11 & 1.94 & 2.63 & 0.154 & 0.108 & 0.008 & 0.001 & 0.000 & 0.000 & 61.7 \\ 
NK & 29 & 92 & 0.23 & 1.95 & 1.96 & 0.379 & 0.223 & 0.097 & 0.011 & 0.016 & 0.000 & 97.3 \\ 
UA & 48 & 158 & 0.14 & 2.03 & 2.23 & 0.346 & 0.138 & 0.122 & 0.003 & 0.034 & 0.000 & 96.0 \\ 
US & 58 & 113 & 0.07 & 2.09 & 3.03 & 0.115 & 0.067 & 0.011 & 0.000 & 0.000 & 0.000 & 35.3 \\ 
WN & 88 & 522 & 0.14 & 1.99 & 2.04 & 0.335 & 0.136 & 0.106 & 0.002 & 0.031 & 0.000 & 99.9 \\ 
\end{tabular}
\caption{Descriptive statistics for eight carrier networks in 2013Q4. The carriers are American Airlines (AA), Alaska Airlines (AS), Delta Air Lines (DL), AirTran Airways (FL), Spirit Airlines (NK), United Airlines (UA), US Airways (US), and Southwest Airlines (WN). The networks are generally small (nodes and edges) and sparse (density). The average path length (apl) is the number of edges in the shortest path between two nodes, averaged over all pairs of nodes: it is close to two for all networks. We compare the average path length and clustering coefficients $C(3)$, $C(4)$ and $C(5)$ to realizations from  Erd{\H{o}}s-R{\'{e}}nyi random graphs $G(n, p)$ with $n$ equal to the number of nodes in the observed network, and edge-formation probability $p$ equal to its density. The random average path length ($\mathrm{apl}_{\mathrm{rand}}$) for $G(n, p)$ is averaged over 1000 replications. Some of the random graphs are not connected (and Connected \% gives the percentage of connected realizations across all replications), and have an infinite average path length. For that reason, we compute the average path length only across connected realizations, with clustering coefficients averaged over all realizations of $G(n, p)$, both connected and disconnected, as $C(3)_\mathrm{rand}$, $C(4)_\mathrm{rand}$ and $C(5)_\mathrm{rand}$. The percentage of connected realizations is positively related to the network density. There is some evidence (strongest for Southwest Airlines) that the airlines have a small-world property, with similar average path lengths to a random graph, but higher three-node clustering; although the average path lengths for connected random graphs are generally higher than those observed in the real networks. Alaska Airlines has some evidence of randomness, with particularly low clustering. We see that generalized clustering $C(4)$ and $C(5)$ are typically higher than random for all carriers.}
\label{tab:network_descriptives}
\end{center}
\end{table}

\subsection{The distribution of maximal cliques}
Since $C(b)$ is intimately related to the relative number of cliques, we use the Bron-Kerbosch algorithm to identify all cliques in a given network. Figure \ref{fig:4_clique_southwest} displays the 2013Q4 network of Southwest Airlines, and we highlight one maximal 4-clique for illustration, between Albuquerque, Dallas, Houston, and Kansas City. We report the airport codes and co-ordinates in Table \ref{tab:list_of_airports_2013_4} in Appendix \ref{sec:correlations}. It is interesting to see how many maximal cliques of any given size there are in a network, and whether this distribution is stable over time. For Southwest's network, in Figure \ref{fig:kernel_plots_maximal_cliques_southwest}, we observe that the distribution is more spread out, and that more larger cliques appear, over time. There is a maximum clique size of eleven, which corresponds to 12.5\% of all of the airports served by Southwest in 2013Q4. This might seem surprising, given that Southwest's network is relatively sparse, with a density $d(G)$ close to 15\% in that quarter. Since every airport in the maximum clique has at least 11 connections, we can think of it as a group of ``important'' airports, that are also very highly connected among themselves.\footnote{We find evidence that nodes that belong to maximal cliques in Southwest's network are more connected, on average, than nodes that are not in maximal cliques, and that the average degree of nodes in maximal cliques increases in the order of the clique (Figure \ref{fig:degree_of_nodes_in_cliques_WN_2013_4}).} An operational reason for developing such groups could be to enable the opening of a large number of new indirect routes between airport pairs, at relatively low cost, with the addition of a few well-chosen direct routes. It seems likely that Southwest, through its network expansion, has focused on increasing the size and connectivity of a moderate number of ``central'' airports while also creating links from non-central airports into this group.\footnote{Not all networks evolve in this way e.g. the distribution of maximal cliques for American (AA) is far more stable over time (Figure \ref{fig:kernel_plots_maximal_cliques_american}).}

\begin{figure}\centering
    \includegraphics[scale=0.45]{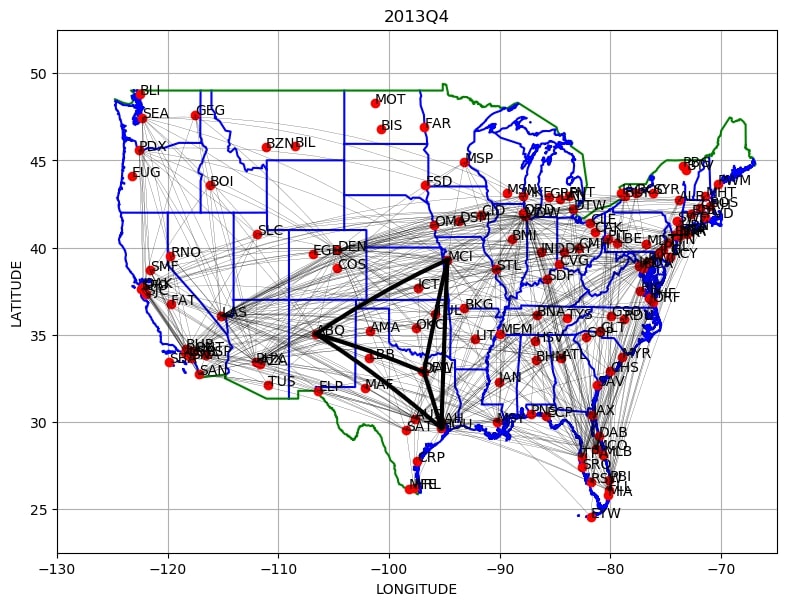}
    \caption{The Albuquerque--Dallas--Houston--Kansas City  maximal 4-clique in Southwest's network, found using the Bron-Kerbosch algorithm, is one example of a completely-connected set of four nodes. Southwest Airlines served every pairwise route among these four airports in 2013Q4. Since this is a maximal clique, no other airport in Southwest's network can be added to the clique while preserving its complete connectivity (to create a 5-clique).}\label{fig:4_clique_southwest}
\end{figure}

\begin{figure}\centering
              \begin{minipage}[b]{0.45\linewidth}
                \centering
                \includegraphics[scale=0.35]{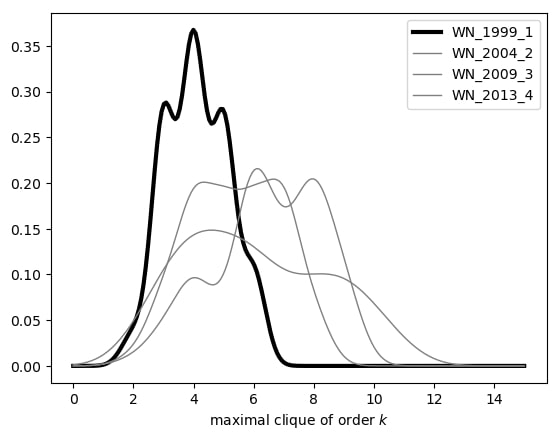}
              \end{minipage}
              \begin{minipage}[b]{0.45\linewidth}
                \centering
                \includegraphics[scale=0.35]{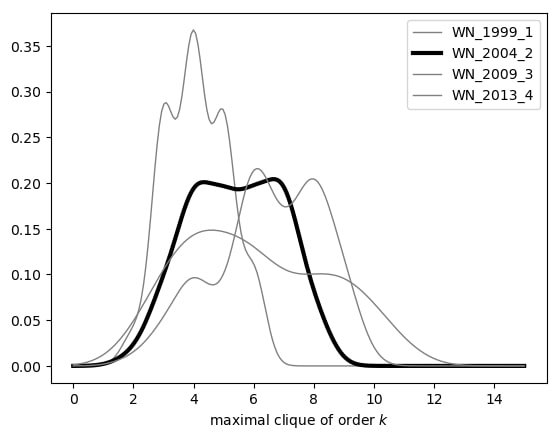}
              \end{minipage} 
              \begin{minipage}[b]{0.45\linewidth}
                \centering
                \includegraphics[scale=0.35]{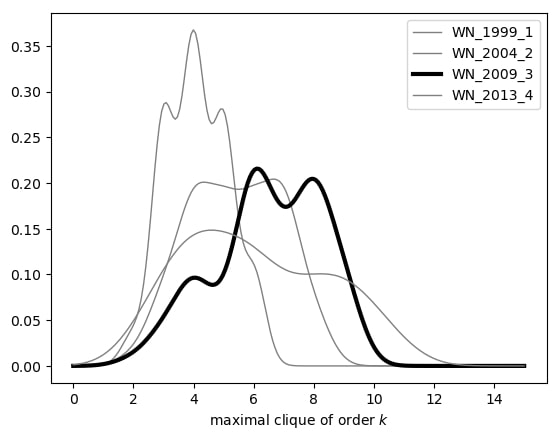}
              \end{minipage}
              \begin{minipage}[b]{0.45\linewidth}
                \centering
                \includegraphics[scale=0.35]{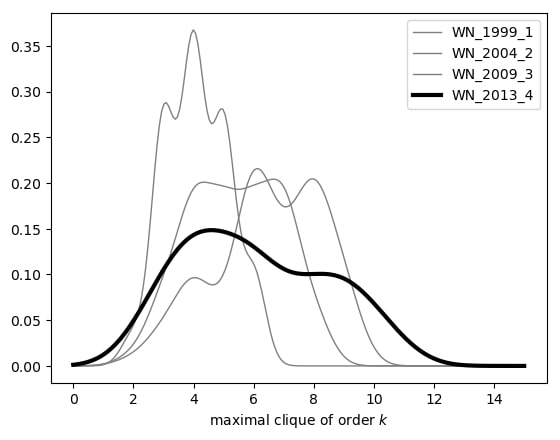}
              \end{minipage} 
              \caption{The distribution of maximal $k$-cliques in Southwest's network, in 1999Q1, 2004Q2, 2009Q3 and 2013Q4. This shows that Southwest is creating increasingly large groups of very highly interconnected airports over time.}\label{fig:kernel_plots_maximal_cliques_southwest}
\end{figure}

\subsection{Dynamic variation in higher-order clustering}\label{sec:test}
Since there are many cliques with more than three nodes, we examine how $C(3)$, $C(4)$ and $C(5)$ vary across carriers, and over time (Figure \ref{eq:cb_dynamics}). We make the following remarks:

\begin{figure}\centering
                  \begin{subfigure}[b]{0.4\linewidth}
                    \centering
                   \includegraphics[scale=0.2]{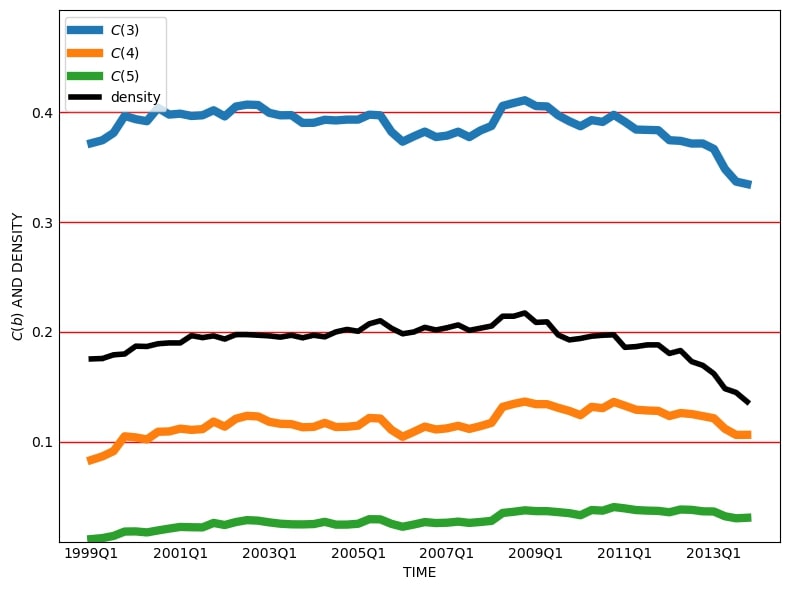}
                    \caption{Southwest Airlines.}
                    \vspace{3mm}
                  \end{subfigure}
                  \begin{subfigure}[b]{0.4\linewidth}
                    \centering
                    \includegraphics[scale=0.2]{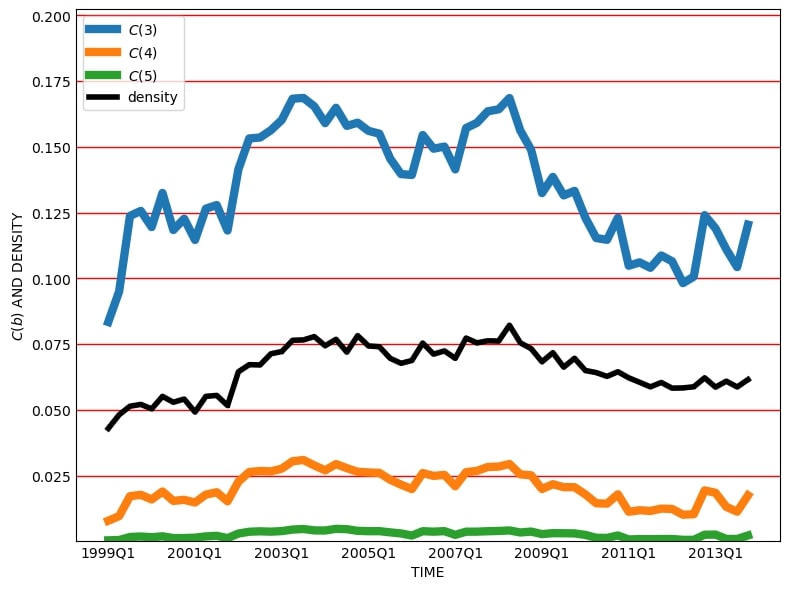}
                    \caption{American Airlines.}
                    \vspace{3mm}
                  \end{subfigure} 
                  \begin{subfigure}[b]{0.4\linewidth}
                    \centering
                    \includegraphics[scale=0.2]{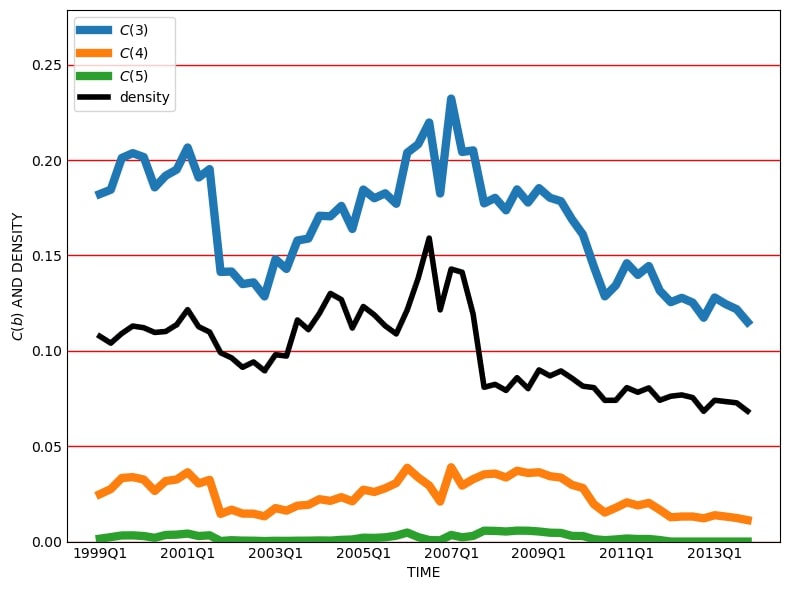}
                    \caption{US Airways.}
                    \vspace{3mm}
                  \end{subfigure}
                  \begin{subfigure}[b]{0.4\linewidth}
                    \centering
                    \includegraphics[scale=0.2]{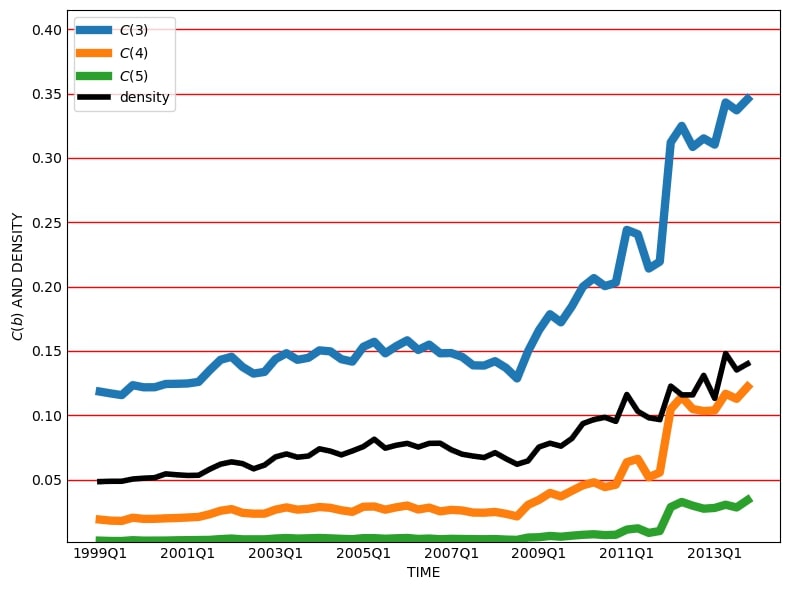}
                    \caption{United Airlines.}
                    \vspace{3mm}
                  \end{subfigure}
                  \begin{subfigure}[b]{0.4\linewidth}
                    \centering
                    \includegraphics[scale=0.2]{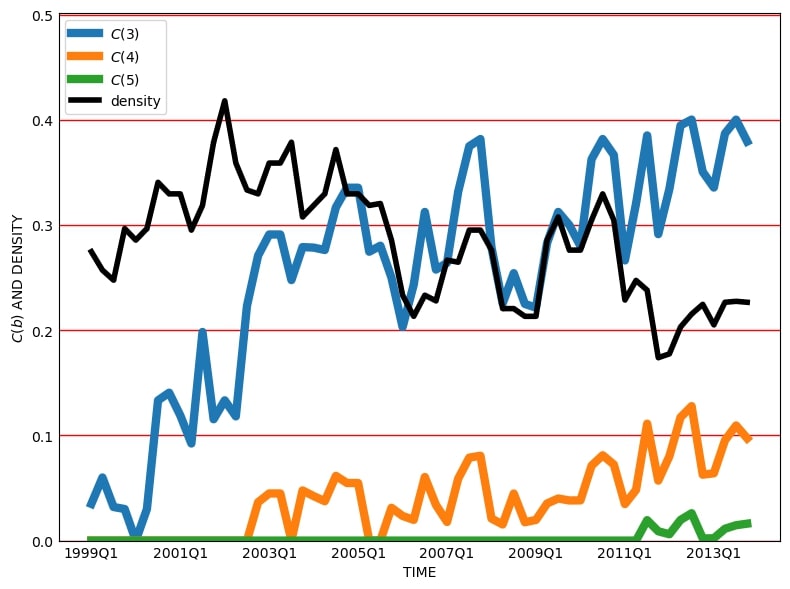}
                    \caption{Spirit Airlines.}
                    \vspace{3mm}
                  \end{subfigure} 
                  \begin{subfigure}[b]{0.4\linewidth}
                    \centering
                    \includegraphics[scale=0.2]{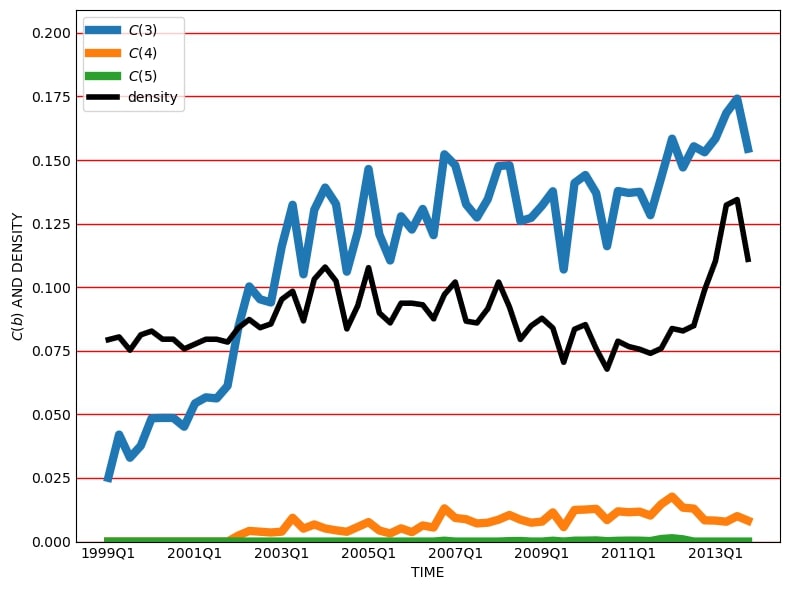}
                    \caption{AirTran Airways.}
                    \vspace{3mm}
                  \end{subfigure} 
                  \begin{subfigure}[b]{0.4\linewidth}
                    \centering
                    \includegraphics[scale=0.2]{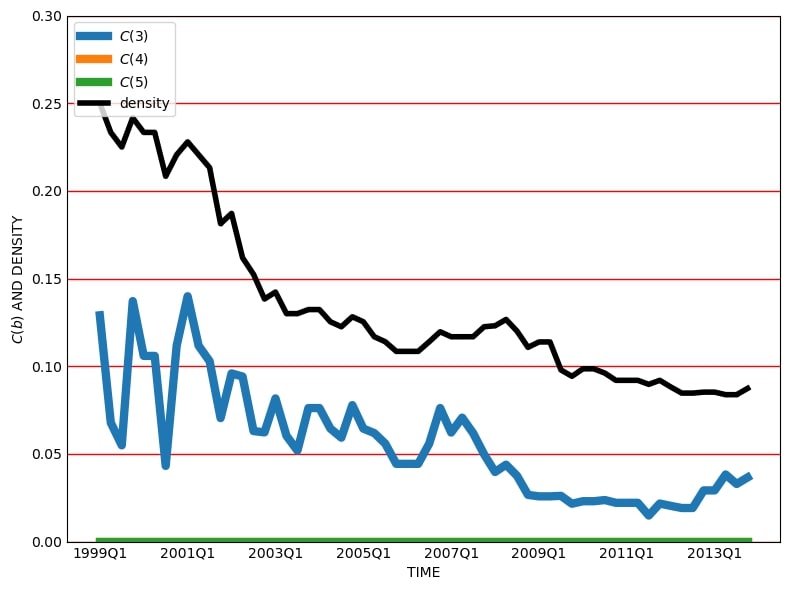}
                    \caption{Alaska Airlines.}
                  \end{subfigure} 
                  \begin{subfigure}[b]{0.4\linewidth}
                    \centering
                    \includegraphics[scale=0.2]{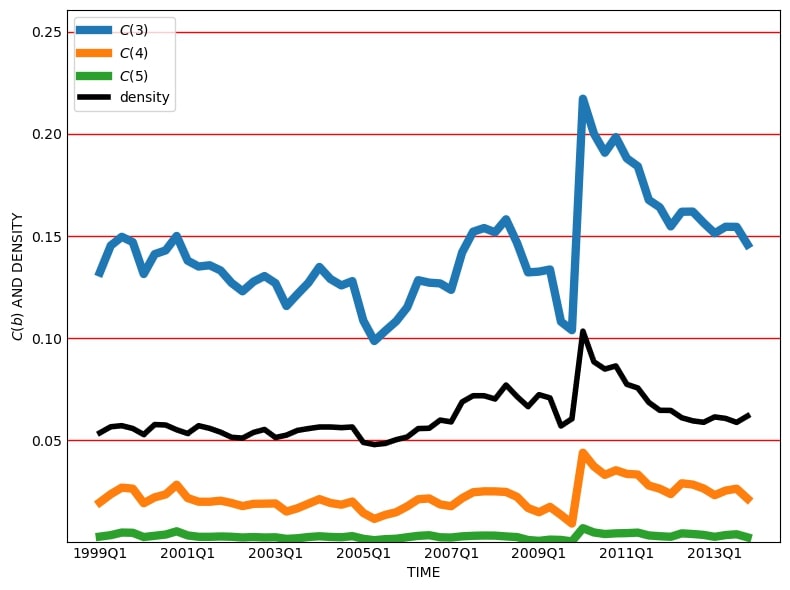}
                    \caption{Delta Air Lines.}
                  \end{subfigure} 
                  
                  \caption{The dynamic behaviour of $C(3)$, $C(4)$, $C(5)$ and density from 1999Q1 to 2013Q4 for eight airline carriers. We observe substantial heterogeneity across carriers, a high degree of correlation between generalized clustering $C(b)$ for different $b$, and evidence that $C(b)$ is decreasing in $b$ (in the text, we give theoretical counterexamples to the latter observation).}\label{eq:cb_dynamics}
\end{figure}
            
\begin{itemize}
    \item There is considerable heterogeneity across carriers. For instance, Southwest has quite a stable $C(3)$ over 1999Q1 to 2013Q4, despite its significant expansion in terms of airports and routes. On the other hand, United has far more clustering (in triangles) from 2009 onwards, while Alaska has progressively less.
    
    \item Generalized clustering $C(b)$ is highly positively correlated across $b$, for some networks e.g. Delta, and is also positively correlated with network density (see Table \ref{tab:correlations} in Appendix \ref{sec:correlations}). Some of this follows by construction e.g. for every newly formed 4-clique in a network, there will be between two and four new 3-cliques, while every newly formed 5-clique will create between two and five new 4-cliques and between three and ten new 3-cliques. High correlation reduces the information-content of $C(4)$ and $C(5)$, but it is unclear whether this result holds for other real-world networks. To control for this correlation, we consider the following regressions: $C(4) = \mathrm{constant} + \beta \, C(3) + \mathrm{error}$ and $C(5) = \mathrm{constant} + \beta \, C(3) + \gamma \, C(4) + \mathrm{error}$, where the residuals could be used rather than $C(4)$ and $C(5)$ themselves. We illustrate the $C(4)$ procedure in Figure \ref{fig:correlation_correction}, for US Airways and Southwest (WN). Since $C(3)$ and $C(4)$ display evidence of a unit root (US) and a unit root and trend (WN), we first run regressions of the form $\Delta \, C(b)_{t} = \alpha + \delta \, t + u(b)_{t}$, for $b=3,4$. We then regress the difference and trend stationary $\widehat{u}(4)$ on a constant and $\widehat{u}(3)$, and find that 76\% (US) and 89\% (WN) of the variation in $C(4)$ is ``explained'' by $C(3)$. In this sense, $C(4)$ is moderately informative once $C(3)$ has been accounted for. It is unclear if other networks will give the same results, but similar behaviour is noticed on $C_{b-1}$ by \cite{yin_etal18} for other real-world networks.
    
    \item Table \ref{tab:network_descriptives} and Figure \ref{eq:cb_dynamics} provide evidence that $C(3) > C(4) > C(5)$, and we might think that this holds quite generally. However, Figure \ref{fig:clustering_difference} shows that this is not necessarily the case for $G(n,p)$. We can also construct a series of counterexamples using the lollipop graph $L(b, n-b)$.
    \begin{figure}\centering
    	\includegraphics[width=.4\linewidth]{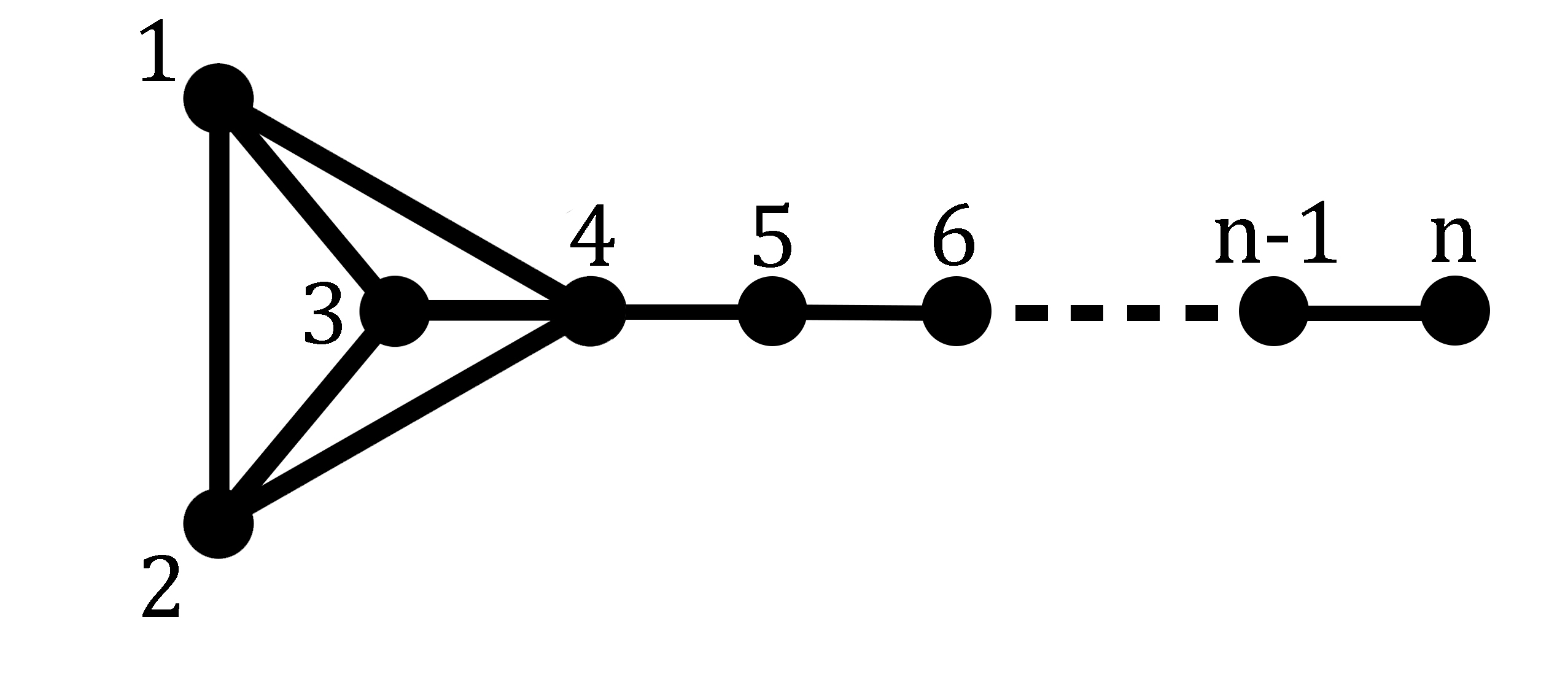}
    	\caption{The lollipop graph $L(4, n-4)$ is a counterexample to $C(3) \geq C(4)$, which does not hold for $n \geq 27$.}
    	\label{fig:counterexample}
    \end{figure}
    As $n$ increases given $b$, the number of complete subgraphs of order no more than $b$ does not change (for instance, there are four triangles and one 4-complete subgraph in the lollipop graph $L(4 ,n-4)$ of Figure \ref{fig:counterexample}). Furthermore, increasing $n$ after a certain point will only add paths of length $b$ to the denominator of $C(b)$, and no other spanning trees (e.g. $b$-stars). For $L(4, n-4)$, we have already seen that
    \begin{equation*}
        C(3) = \frac{12}{n + 10}; \, n \geq 5; \quad C(4) = \frac{16}{n + 22}; \, n \geq 6,
    \end{equation*}
    and it follows that $C(3) \geq C(4)$ as $n \leq 26$, with equality when $C(3) = C(4) = 1/3$.
    \begin{figure}\centering
    	\includegraphics[width=.4\linewidth]{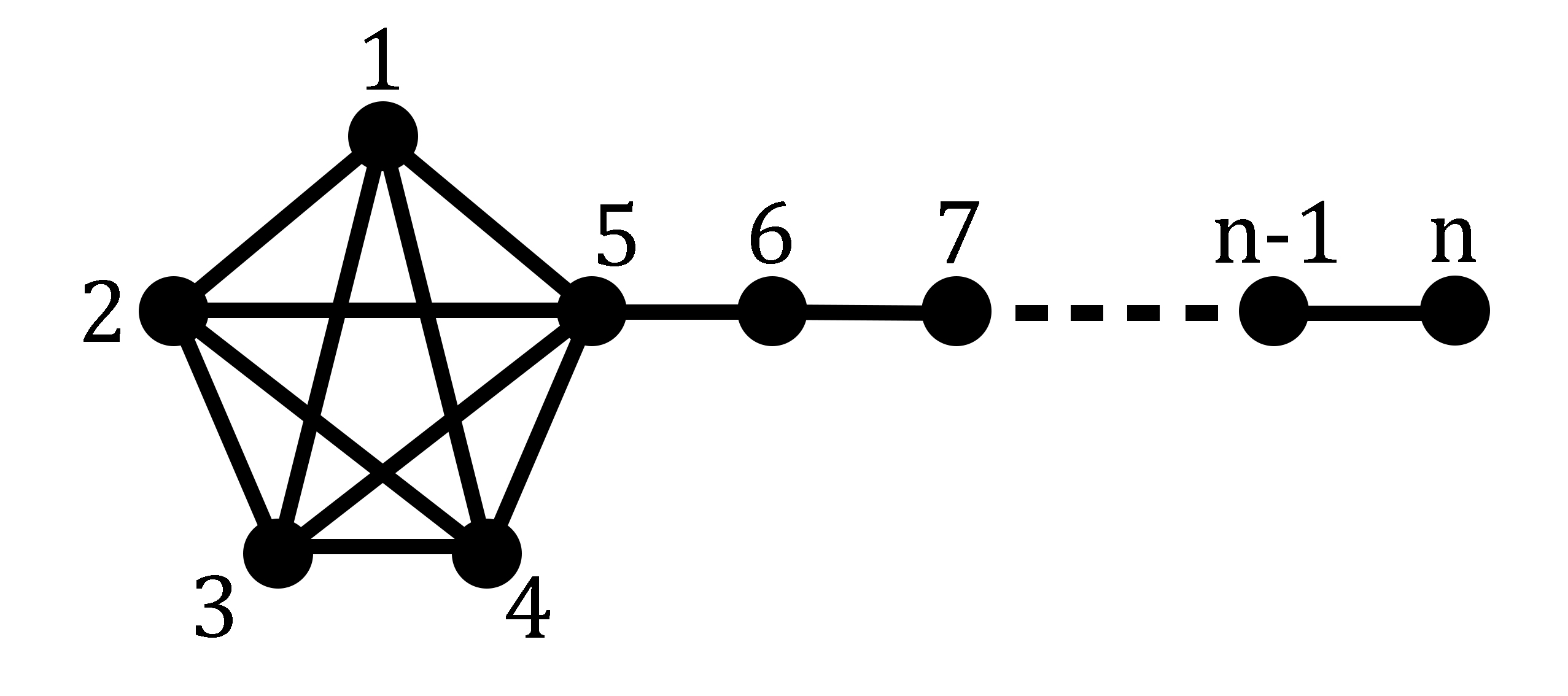}
    	\caption{The lollipop graph $L(5, n-5)$ is a counterexample to $C(4) \geq C(5)$, which does not hold for $n \geq 98$.}
    	\label{fig:counterexample2}
    \end{figure}
    In the lollipop $L(5, n-5)$ of Figure \ref{fig:counterexample2}, the number of 4-complete and 5-complete subgraphs is constant as $n$ increases. Beyond a certain point, only 4-paths and 5-paths are added to the denominators of $C(4)$ and $C(5)$ and no further 4-stars or 5-stars or 5-arrows are created. We can show that

    \begin{equation*}
        C(4) = \frac{80}{n + 95}; \, n \geq 7; \quad C(5) = \frac{125}{n + 203}; \, n \geq 8,
    \end{equation*}
    from which $C(4) \geq C(5)$ as $n \leq 97$. Equality occurs when $C(4) = C(5) = 5/12$. Incidentally, for this graph, $C(3) \geq C(4)$ as $n \leq 12.2$, i.e., $n < 13$. We could use this construction to show that $C(b) < C(b+1)$ for any $b < n$ and sufficiently large $n$.

\end{itemize}

\newpage

\section{Conclusions}
    
        We have proposed a generalized clustering coefficient $C(b)$ of order greater than or equal to three, that nests the usual overall clustering, or transitivity, $C(3)$. We investigate the properties of $C(b)$ on random and small-world graphs, and propose an algorithm for implementation that is based on analytical subgraph counts, and that is practical for $b=3,4,5$ when the graph is not too large. Our work complements the recent paper by Yin-Benson-Leskovec \cite{yin_etal18}, who generalize overall clustering in a different way, and we draw careful comparisons with their findings. We illustrate the performance of our measure using data on U.S. airline route networks, and provide new insight into the strategic behaviour that leads carriers to develop small groups of highly connected airports. Extending these ideas to generalize the notion of a ``hub'' node to multi-node hubs, in both transportation and other real-world networks, is a promising avenue for future work. We expect analytic formulae for subgraph enumeration to have application in other areas of applied graph theory (e.g. motif detection \cite{agasse-duval_lawford18}), although it currently seems too difficult to derive analytic formulae for all subgraphs on more than five nodes (for complete results on five node subgraphs, see \cite{lawford20}). Future work linking graphs and econometrics should also lead to a better understanding of the economic, strategic and spatial factors that drive dynamic clustering in real-world networks e.g. \cite{depaula17, depaula20, graham_depaula20} suggest possible applications in economics to game-theoretic network formation models and production networks and, in social networks, to the coordination of clustered individuals on collective actions.

\newpage

\appendix

\setcounter{table}{0}
\renewcommand{\thetable}{B.\arabic{table}}

\counterwithin{figure}{section}

\section{Proofs}\label{sec:proofs}
Derivations for the 3-star $M_{3}^{(3)}$, the triangle $M_{7}^{(3)}$, the 4-star $M_{11}^{(4)}$, the 4-path $M_{13}^{(4)}$, the tadpole $M_{15}^{(4)}$, and the 4-complete $M_{63}^{(4)}$ are given in \cite[Proposition 2.1]{agasse-duval_lawford18}. For completeness, we repeat the results here, without proof, in Proposition \ref{thm:analytic_count}. We also include the three spanning trees on five nodes: the 5-star $M_{75}^{(5)}$, the 5-arrow $M_{77}^{(5)}$ and the 5-path $M_{86}^{(5)}$, as well as the 5-complete $M_{1023}^{(5)}$, all with their corresponding proofs.

\begin{proposition}[Analytic formulae for nested subgraph enumeration]\label{thm:analytic_count}
\begin{equation*}\label{eq:m_3_3}
|M_{3}^{(3)}| = \sum_{i}\dbinom{k_{i}}{2} = \frac{1}{2} \, \sum_{i}k_{i}(k_{i}-1).
\end{equation*}
\begin{equation}\label{eq:m_7_3}
|M_{7}^{(3)}| = \frac{1}{6} \, \tr \, (g^{3}).
\end{equation}
\begin{equation*}\label{eq:m_11_4}
|M_{11}^{(4)}| = \sum_{i}\dbinom{k_{i}}{3} = \frac{1}{6} \, \sum_{i}k_{i}(k_{i}-1)(k_{i}-2).
\end{equation*}
\begin{equation*}\label{eq:m_13_4}
|M_{13}^{(4)}| = \sum_{(i, j) \in E}(k_{i}-1)(k_{j}-1) - 3 \, |M_{7}^{(3)}|.
\end{equation*}
\begin{equation*}\label{eq:m_15_4}
|M_{15}^{(4)}| = \frac{1}{2} \sum_{k_{i} > 2}(g^{3})_{ii} \, (k_{i}-2).
\end{equation*}
\begin{equation}\label{eq:m_63_4}
|M_{63}^{(4)}| = \frac{1}{24} \sum_{i} \tr(g_{-i}^{3}).
\end{equation}
\begin{equation}\label{eq:m_75_5}
|M_{75}^{(5)}| = \sum_{i}\dbinom{k_{i}}{4} = \frac{1}{24} \, \sum_{i}k_{i}(k_{i}-1)(k_{i}-2)(k_{i}-3).
\end{equation}
\begin{equation}\label{eq:m_77_5}
|M_{77}^{(5)}| = \sum_{(i,j)^{\star} \in E} \dbinom{k_{i}-1}{2}(k_{j}-1) - 2 |M_{15}^{(4)}|.
\end{equation}
\begin{equation}\label{eq:m_86_5}
|M_{86}^{(5)}| = \frac{1}{2}\sum_{i \neq j} (g^{4})_{ij} - 2\,|M_{3}^{(3)}| - 9 \, |M_{7}^{(3)}| - 3 \, |M_{11}^{(4)}| - 2 \, |M_{13}^{(4)}| - 2 \, |M_{15}^{(4)}|.
\end{equation}
\begin{equation}\label{eq:m_1023_5}
|M_{1023}^{(5)}| = \frac{1}{5} \, \sum_{i} |M_{63}^{(4)} (g_{-i}) | = \frac{1}{120} \sum_{i} \sum_{j \in \Gamma_{G}(i)}\tr(((g_{-i})_{-j})^{3}) .
\end{equation}
\end{proposition}

\begin{remark}
In (\ref{eq:m_63_4}), $g_{-i}$ is the adjacency matrix corresponding to the subgraph induced by the neighbourhood $\Gamma_{G}(i)$ of $i$, which we denote by $G_{-i} = (V(\Gamma_{G}(i)), E(\Gamma_{G}(i)))$, and we use (\ref{eq:m_7_3}) to count the number of triangles.
\end{remark}

\begin{remark}
In (\ref{eq:m_77_5}), $\sum_{(i,j)^{*} \in E}$ denotes summation over all edges in $E$, in \emph{both} directions $(i,j)$ and $(j,i)$.
\end{remark}

\begin{remark}
In (\ref{eq:m_1023_5}), $(g_{-i})_{-j}$ is the adjacency matrix corresponding to the subgraph induced by the neighbourhood $\Gamma_{G-i}(j)$ of $j$, which we denote by $G_{-i-j} = (V(\Gamma_{G-i}(j)), E(\Gamma_{G-i}(j)))$, and we use (\ref{eq:m_63_4}) to count the number of 4-cliques.
\end{remark}
\begin{proof}[Proof of Proposition \ref{thm:analytic_count}] We treat each subgraph separately, and only report proofs that are not presented in \cite[Proposition 2.1]{agasse-duval_lawford18}.

\begin{enumerate}[label=(\alph*)]

    \item $|M_{75}^{(5)}|$: Node $i$ has edges to $k_{i}$ neighbours, and any four of those edges will form a 5-star, centered on $i$. The result (\ref{eq:m_75_5}) follows immediately.
    
    \item $|M_{77}^{(5)}|$: The method of proof is similar to that used for the count of the nested 4-path $|M_{13}^{(4)}|$ in \cite{agasse-duval_lawford18}. Consider any edge $(i, j) \in E$, as the central edge in a 5-arrow. Let $i$ and $j$ have degrees three and two respectively, and let node $i$ be directly-connected to nodes $x$ and $z$, and let node $j$ be directly-connected to node $y$. Node $i$ has $k_{i}-1$ possible neighbours (for nodes $x$ and $z$) and node $j$ has $k_{j} - 1$ possible neighbours (for node $y$). There are $\frac{1}{2}(k_{i}-1)(k_{i}-2)(k_{j}-1)$ ways in which two neighbours of $i$ can be paired with a neighbour of $j$, which gives a total of $\sum_{(i,j)^{\star} \in E} \dbinom{k_{i}-1}{2}(k_{j}-1)$ across all possible central edges, in both directions (we use $(i,j)^{\star}$ to denote ``$(i,j)$ and $(j,i)$''). This sum includes the unwanted cases $x = y$ and $y=z$, both of which form a tadpole. Since two of the four edges of the tadpole can be a candidate central edge $(i, j)$ of a 5-arrow, we subtract $2 \, |M_{15}^{(4)}|$ to give result (\ref{eq:m_77_5}).
    
    \item $|M_{86}^{(5)}|$: A very similar but less transparent proof can be found in \cite{movarraei_shikare14}. A 5-path is a walk of length 4 with no repeated nodes. Note that $\frac{1}{2}\sum_{i \neq j} (g^{4})_{ij} $ gives the number of walks of length 4 from $i$ to $j$, which does not only include 5-paths. There are five subgraphs in which we can find walks of length 4 that are not 5-paths:
 
 \begin{center}
    
      \begin{footnotesize}
                \begin{tabular*}{\linewidth}{@{\extracolsep{\fill}}l*{5}{c}}
                    \toprule
                    & \multicolumn{5}{c}{Subgraph}\\
                    \cline{2-6}
                    & 3-path & triangle & 4-star & 4-path & tadpole \\
                    
                    & $M_{3}^{(3)}$ & $M_{7}^{(3)}$ & $M_{11}^{(4)}$ & $M_{13}^{(4)}$ & $M_{15}^{(4)}$  \\
                    \midrule
                    Number of other walks of length 4 & 2 & 9 & 3 & 2 & 2 \\
                    \bottomrule
                \end{tabular*}
            \end{footnotesize}

 \end{center}

So, by removing them from the sum, we have (\ref{eq:m_86_5}) as required.
    
    \item $|M_{1023}^{(5)}|$: Consider a 4-complete subgraph $M_{63}^{(4)}$ comprised of nodes $j$, $k$, $\ell$ and $m$. Let each node be in the neighbourhood $\Gamma_{G}(i)$ of some node $i$ such that $i \ne j \ne k \ne \ell \ne m$. Hence, the five nodes $i$, $j$, $k$, $\ell$ and $m$, and the edges between them, form a 5-complete subgraph $M_{1023}^{(5)}$. The quantity $|M_{63}^{(4)} (g_{-i})|$ gives the number of 5-complete subgraphs that contain node $i$, where $g_{-i}$ is the adjacency matrix corresponding to the subgraph induced by $\Gamma_{G}(i)$. By symmetry, summing across all nodes $i$ will give five times the total count of 5-complete subgraphs in the graph, and so we divide the sum by five to give result (\ref{eq:m_1023_5}), which can be simplified further by using (\ref{eq:m_63_4}) to count 4-complete subgraphs in each subgraph $G_{-i}$.

\end{enumerate}

\end{proof}

\begin{proof}[Proof of Proposition \ref{thm:case_c(b)=1}]
We consider the ``if'' and ``only if'' parts separately:
\begin{itemize}

    \item (\textbf{if}) Let $G$ be complete. Hence, each set of $b$ nodes of $G$ forms a $b$-clique and $G$ contains exactly $\binom{n}{b}$ $b$-cliques. 
    The number of $b$-spanning trees of $G$ is equal to the number of $b$-spanning trees enclosed in any $b$-clique which is, using Cayley’s formula:
    \begin{equation*}\label{eq:cayleydenom}
    b^{b-2} \, \times \, \binom{n}{b},
    \end{equation*}
    from which (\ref{eq:c_b}) gives $C(b)=1$.

    \item (\textbf{only if}) We prove this part by contrapositive. Suppose that $G$ is not complete.
	Since $G$ has at least $b$ nodes, we can find a connected subgraph $G’$ of $G$ with $b$ nodes such that $G’$ is not a $b$-clique, and we can extract a $b$-spanning tree from $G’$ by removing any cycles. Hence, there is at least one $b$-spanning tree in $G$ which is not enclosed in a $b$-clique. It follows that:
	
    \begin{equation*}
        \begin{aligned}
            \textrm{number of $b$-spanning trees in $G$} & \geq \textrm{number of $b$-spanning trees enclosed in a $b$-clique} + 1 \\
            & > \textrm{number of $b$-spanning trees enclosed in a $b$-clique}\\
            & = b^{b-2} \times \textrm{number of $b$-cliques in $G$},
        \end{aligned}
    \end{equation*}
    and so $C(b) < 1$ from (\ref{eq:c_b}), which proves the proposition.
    \end{itemize}

\end{proof}

\newpage

\section{Additional Figures and Tables}\label{sec:correlations}

\begin{table}
\scriptsize
\begin{center}
\begin{tabular}{lrrclrrclrr}
code & x & y &  & code & x & y &  & code & x & y\\ \hline
 & & & & & & & & & & \\ 
ABQ & -106.61 & 35.04 &  & ACY & -74.58 & 39.47 &  & ALB & -73.80 & 42.73\\ 
AMA & -101.71 & 35.23 &  & ATL & -84.43 & 33.64 &  & AUS & -97.67 & 30.19\\ 
AZA & -111.66 & 33.31 &  & BDL & -72.68 & 41.94 &  & BHM & -86.75 & 33.57\\ 
BIL & -108.53 & 45.80 &  & BIS & -100.75 & 46.78 &  & BKG & -93.20 & 36.53\\ 
BLI & -122.53 & 48.80 &  & BMI & -88.92 & 40.48 &  & BNA & -86.68 & 36.12\\ 
BOI & -116.22 & 43.56 &  & BOS & -71.00 & 42.36 &  & BTV & -73.15 & 44.47\\ 
BUF & -78.73 & 42.94 &  & BUR & -118.35 & 34.20 &  & BWI & -76.67 & 39.18\\ 
BZN & -111.15 & 45.78 &  & CAK & -81.44 & 40.92 &  & CHS & -80.04 & 32.90\\ 
CID & -91.71 & 41.88 &  & CLE & -81.85 & 41.42 &  & CLT & -80.93 & 35.22\\ 
CMH & -82.88 & 40.00 &  & COS & -104.72 & 38.82 &  & CRP & -97.50 & 27.77\\ 
CVG & -84.67 & 39.05 &  & DAB & -81.05 & 29.18 &  & DAL & -96.85 & 32.85\\ 
DAY & -84.18 & 39.75 &  & DCA & -77.04 & 38.85 &  & DEN & -104.67 & 39.86\\ 
DFW & -97.04 & 32.90 &  & DSM & -93.66 & 41.53 &  & DTW & -83.35 & 42.21\\ 
ECP & -85.80 & 30.36 &  & EGE & -106.92 & 39.63 &  & ELP & -106.38 & 31.80\\ 
EUG & -123.22 & 44.12 &  & EWR & -74.17 & 40.69 &  & EYW & -81.77 & 24.55\\ 
FAR & -96.82 & 46.92 &  & FAT & -119.72 & 36.77 &  & FLL & -80.15 & 26.07\\ 
FNT & -83.74 & 42.97 &  & FSD & -96.74 & 43.58 &  & GEG & -117.53 & 47.62\\ 
GRR & -85.53 & 42.88 &  & GSO & -79.94 & 36.10 &  & GSP & -82.22 & 34.90\\ 
HOU & -95.28 & 29.65 &  & HPN & -73.70 & 41.07 &  & HRL & -97.75 & 26.20\\ 
HSV & -86.78 & 34.64 &  & IAD & -77.46 & 38.94 &  & IAG & -79.03 & 43.10\\ 
IAH & -95.34 & 29.98 &  & ICT & -97.43 & 37.65 &  & ILG & -75.60 & 39.68\\ 
IND & -86.29 & 39.72 &  & ISP & -73.10 & 40.80 &  & JAN & -90.08 & 32.31\\ 
JAX & -81.63 & 30.42 &  & JFK & -73.78 & 40.63 &  & LAN & -84.58 & 42.78\\ 
LAS & -115.17 & 36.08 &  & LAX & -118.41 & 33.94 &  & LBB & -101.83 & 33.67\\ 
LBE & -79.40 & 40.28 &  & LGA & -73.87 & 40.77 &  & LGB & -118.15 & 33.82\\ 
LIT & -92.22 & 34.73 &  & MAF & -102.20 & 31.94 &  & MCI & -94.73 & 39.29\\ 
MCO & -81.31 & 28.43 &  & MDT & -76.76 & 40.19 &  & MDW & -87.75 & 41.78\\ 
MEM & -89.97 & 35.07 &  & MFE & -98.24 & 26.18 &  & MHT & -71.44 & 42.93\\ 
MIA & -80.27 & 25.78 &  & MKE & -87.90 & 42.95 &  & MLB & -80.63 & 28.10\\ 
MOT & -101.28 & 48.27 &  & MSN & -89.34 & 43.14 &  & MSP & -93.22 & 44.88\\ 
MSY & -90.26 & 29.99 &  & MYR & -78.97 & 33.70 &  & OAK & -122.22 & 37.72\\ 
OKC & -97.60 & 35.39 &  & OMA & -95.90 & 41.30 &  & ONT & -117.60 & 34.06\\ 
ORD & -87.90 & 41.98 &  & ORF & -76.20 & 36.90 &  & ORH & -71.88 & 42.27\\ 
PBG & -73.47 & 44.65 &  & PBI & -80.10 & 26.68 &  & PDX & -122.60 & 45.59\\ 
PHF & -76.50 & 37.13 &  & PHL & -75.24 & 39.87 &  & PHX & -112.03 & 33.43\\ 
PIT & -80.23 & 40.49 &  & PNS & -87.18 & 30.47 &  & PSP & -116.50 & 33.83\\ 
PVD & -71.43 & 41.73 &  & PWM & -70.30 & 43.65 &  & RDU & -78.79 & 35.88\\ 
RIC & -77.32 & 37.50 &  & RNO & -119.77 & 39.50 &  & ROC & -77.67 & 43.12\\ 
RSW & -81.76 & 26.54 &  & SAN & -117.18 & 32.73 &  & SAT & -98.47 & 29.53\\ 
SAV & -81.20 & 32.13 &  & SBA & -119.84 & 33.43 &  & SDF & -85.74 & 38.17\\ 
SEA & -122.31 & 47.45 &  & SFO & -122.38 & 37.62 &  & SJC & -121.92 & 37.35\\ 
SLC & -111.97 & 40.79 &  & SMF & -121.62 & 38.70 &  & SNA & -117.87 & 33.67\\ 
SRQ & -82.55 & 27.40 &  & STL & -90.37 & 38.75 &  & SWF & -74.02 & 41.50\\ 
SYR & -76.12 & 43.12 &  & TPA & -82.53 & 27.98 &  & TTN & -74.81 & 40.28\\ 
TUL & -95.89 & 36.20 &  & TUS & -110.94 & 32.12 &  & TYS & -83.92 & 35.95\\ 
\end{tabular}
\caption{List of airports by IATA code in 2013Q4, all carriers,                 with their longitude (x) and latitude (y). The identity of the airport corresponding to each                 IATA code can be found at http://www.iata.org/en/publications/directories/code-search/.}
\label{tab:list_of_airports_2013_4}
\end{center}
\end{table}

    \begin{center}
         \begin{table}[!ht]
         \begin{scriptsize}
            \begin{subtable}[b]{0.5\linewidth}
                \centering
                 \begin{tabular*}{\linewidth}{@{\extracolsep{\fill}}l*{9}{c}}
                            \cline{1-5}
                            Variable & $C(3)$ & $C(4)$ & $C(5)$ & $density$ \\
                            \cline{2-5}
                            $C(3)$  & 1.000 & 0.394 & $-0.012$ & 0.790 \\
                            $p$-value  & 0.000 & 0.002 & 0.927 & 0.000 \\
                            \cline{2-5}
                            $C(4)$  & - & 1.000 & 0.910 & 0.365\\
                            $p$-value  & - & 0.000 & 0.000 & 0.004 \\
                            \cline{2-5}
                            $C(5)$  & - & - & 1.000 & 0.039 \\
                            $p$-value  & - & - & 0.000 & 0.766 \\
                            \cline{2-5}
                            $density$  & - & - & - & 1.000 \\
                            $p$-value  & - & - & - & 0.000 \\
                            \cline{1-5}
                    \end{tabular*}
                    \caption{Southwest Airlines.}
                    \vspace{3mm}
              \end{subtable} 
              \begin{subtable}[b]{0.5\linewidth}
                \centering
                \begin{tabular*}{\linewidth}{@{\extracolsep{\fill}}l*{9}{c}}
                            \cline{1-5}
                            Variable & $C(3)$ & $C(4)$ & $C(5)$ & $density$ \\
                            \cline{2-5}
                            $C(3)$  & 1.000 & 0.992 & 0.963 & 0.864 \\
                            $p$-value  & 0.000 & 0.000 & 0.000 & 0.000 \\
                            \cline{2-5}
                            $C(4)$  & - & 1.000 & 0.984 & 0.852 \\
                            $p$-value  & - & 0.000 & 0.000 & 0.000 \\
                            \cline{2-5}
                            $C(5)$  & - & - & 1.000 & 0.861 \\
                            $p$-value  & - & - & 0.000 & 0.000 \\
                            \cline{2-5}
                            $density$  & - & - & - & 1.000 \\
                            $p$-value  & - & - & - & 0.000 \\
                            \cline{1-5}
                    \end{tabular*}
                    \caption{American Airlines.}
                    \vspace{3mm}
              \end{subtable}
               \begin{subtable}[b]{0.5\linewidth}
                \centering
                 \begin{tabular*}{\linewidth}{@{\extracolsep{\fill}}l*{9}{c}}
                            \cline{1-5}
                            Variable & $C(3)$ & $C(4)$ & $C(5)$ & $density$ \\
                            \cline{2-5}
                            $C(3)$  & 1.000 & 0.897 & 0.659 & 0.781 \\
                            $p$-value  & 0.000 & 0.000 & 0.000 & 0.000 \\
                            \cline{2-5}
                            $C(4)$  & - & 1.000 & 0.916 & 0.459\\
                            $p$-value  & - & 0.000 & 0.000 & 0.000 \\
                            \cline{2-5}
                            $C(5)$  & - & - & 1.000 & 0.102 \\
                            $p$-value  & - & - & 0.000 & 0.438 \\
                            \cline{2-5}
                            $density$  & - & - & - & 1.000 \\
                            $p$-value  & - & - & - & 0.000 \\
                            \cline{1-5}
                    \end{tabular*}
                    \caption{US Airways.}
                    \vspace{3mm}
              \end{subtable} 
               \begin{subtable}[b]{0.5\linewidth}
                \centering
                 \begin{tabular*}{\linewidth}{@{\extracolsep{\fill}}l*{9}{c}}
                            \cline{1-5}
                            Variable & $C(3)$ & $C(4)$ & $C(5)$ & $density$ \\
                            \cline{2-5}
                            $C(3)$  & 1.000 & 0.994 & 0.968 & 0.965 \\
                            $p$-value  & 0.000 & 0.000 & 0.000 & 0.000 \\
                            \cline{2-5}
                            $C(4)$  & - & 1.000 & 0.989 & 0.936\\
                            $p$-value  & - & 0.000 & 0.000 & 0.000 \\
                            \cline{2-5}
                            $C(5)$  & - & - & 1.000 & 0.887 \\
                            $p$-value  & - & - & 0.000 & 0.000 \\
                            \cline{2-5}
                            $density$  & - & - & - & 1.000 \\
                            $p$-value  & - & - & - & 0.000 \\
                            \cline{1-5}
                    \end{tabular*}
                    \caption{United Airlines.}
                    \vspace{3mm}
              \end{subtable} 
              \begin{subtable}[b]{0.5\linewidth}
                \centering
                 \begin{tabular*}{\linewidth}{@{\extracolsep{\fill}}l*{9}{c}}
                            \cline{1-5}
                            Variable & $C(3)$ & $C(4)$ & $C(5)$ & $density$ \\
                            \cline{2-5}
                            $C(3)$  & 1.000 & 0.844 & 0.439 & $-0.256$ \\
                            $p$-value  & 0.000 & 0.000 & 0.000 & 0.049 \\
                            \cline{2-5}
                            $C(4)$  & - & 1.000 & 0.709 & $-0.414$\\
                            $p$-value  & - & 0.000 & 0.000 & 0.001 \\
                            \cline{2-5}
                            $C(5)$  & - & - & 1.000 & $-0.450$ \\
                            $p$-value  & - & - & 0.000 & 0.000 \\
                            \cline{2-5}
                            $density$  & - & - & - & 1.000 \\
                            $p$-value  & - & - & - & 0.000 \\
                            \cline{1-5}
                    \end{tabular*}
                    \caption{Spirit Airlines.}
                    \vspace{3mm}
              \end{subtable} 
              \begin{subtable}[b]{0.5\linewidth}
                \centering
                 \begin{tabular*}{\linewidth}{@{\extracolsep{\fill}}l*{9}{c}}
                            \cline{1-5}
                            Variable & $C(3)$ & $C(4)$ & $C(5)$ & $density$ \\
                            \cline{2-5}
                            $C(3)$  & 1.000 & 0.854 & 0.333 & 0.539 \\
                            $p$-value  & 0.000 & 0.000 & 0.009 & 0.000 \\
                            \cline{2-5}
                            $C(4)$  & - & 1.000 & 0.682 & 0.156 \\
                            $p$-value  & - & 0.000 & 0.000 & 0.235 \\
                            \cline{2-5}
                            $C(5)$  & - & - & 1.000 & $-0.255$ \\
                            $p$-value  & - & - & 0.000 & 0.049 \\
                            \cline{2-5}
                            $density$  & - & - & - & 1.000 \\
                            $p$-value  & - & - & - & 0.000 \\
                            \cline{1-5}
                    \end{tabular*}
                    \caption{AirTran Airways.}
                    \vspace{3mm}
              \end{subtable} 
              \begin{subtable}[b]{0.5\linewidth}
                \centering
                 \begin{tabular*}{\linewidth}{@{\extracolsep{\fill}}l*{9}{c}}
                            \cline{1-5}
                            Variable & $C(3)$ & $C(4)$ & $C(5)$ & $density$ \\
                            \cline{2-5}
                            $C(3)$  & 1.000 & NA & NA & 0.840 \\
                            $p$-value  & 0.000 & NA & NA & 0.000 \\
                            \cline{2-5}
                            $C(4)$  & - & 1.000 & NA & NA \\
                            $p$-value  & - & 0.000 & NA & NA \\
                            \cline{2-5}
                            $C(5)$  & - & - & 1.000 & NA \\
                            $p$-value  & - & - & 0.000 & NA \\
                            \cline{2-5}
                            $density$  & - & - & - & 1.000 \\
                            $p$-value  & - & - & - & 0.000 \\
                            \cline{1-5}
                    \end{tabular*}
                    \caption{Alaska Airlines.}
                    \vspace{3mm}
              \end{subtable} 
              \begin{subtable}[b]{0.5\linewidth}
                \centering
                 \begin{tabular*}{\linewidth}{@{\extracolsep{\fill}}l*{9}{c}}
                            \cline{1-5}
                            Variable & $C(3)$ & $C(4)$ & $C(5)$ & $density$ \\
                            \cline{2-5}
                            $C(3)$  & 1.000 & 0.970 & 0.807 & 0.838 \\
                            $p$-value  & 0.000 & 0.000 & 0.000 & 0.000 \\
                            \cline{2-5}
                            $C(4)$  & - & 1.000 & 0.919 & 0.743 \\
                            $p$-value  & - & 0.000 & 0.000 & 0.000 \\
                            \cline{2-5}
                            $C(5)$  & - & - & 1.000 & 0.488 \\
                            $p$-value  & - & - & 0.000 & 0.000 \\
                            \cline{2-5}
                            $density$  & - & - & - & 1.000 \\
                            $p$-value  & - & - & - & 0.000 \\
                            \cline{1-5}
                    \end{tabular*}
                    \caption{Delta Air Lines.}
                    \vspace{3mm}
              \end{subtable} 
            \end{scriptsize}
        \caption{Pearson's correlation test for $C(3)$, $C(4)$, $C(5)$ and density, for different networks. The usual clustering coefficient $C(3)$ is often highly correlated with the network density, but this does not always carry over to generalized clustering $C(4)$ and $C(5)$ (Southwest Airlines, US Airways, Spirit Airlines, AirTran Airways). Furthermore, $C(b)$ is often highly correlated with $C(b-1)$ for $b=4, 5$, although there is less observed correlation between $C(3)$ and $C(5)$.}\label{tab:correlations}
        \end{table} 
        
    \end{center}

\clearpage

\begin{figure}\centering
    \begin{subfigure}[b]{0.75\linewidth}
            \centering
            \includegraphics[scale=0.45]{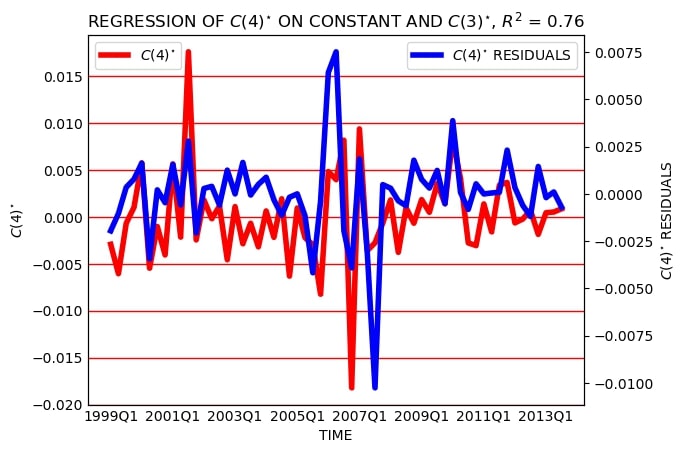}
            \caption{US Airways.}
            \vspace{3mm}
    \end{subfigure}
    \begin{subfigure}[b]{0.75\linewidth}
            \centering
            \includegraphics[scale=0.45]{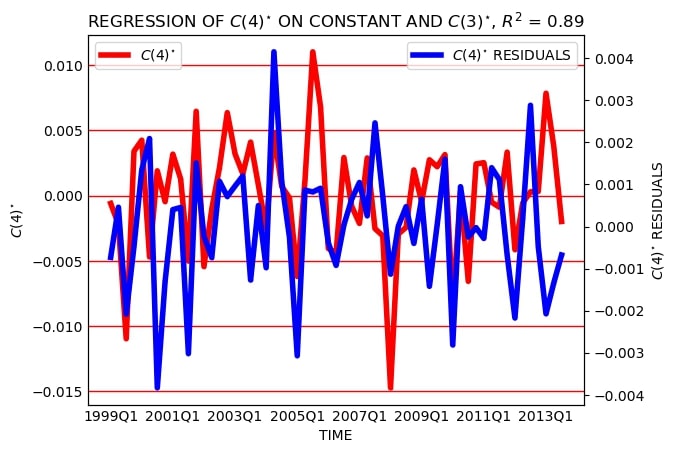}
            \caption{Southwest Airlines.}
            \vspace{3mm}
    \end{subfigure}
    \caption{Results of regressions of $C(4)^{\star}$ on a constant and $C(3)^{\star}$, for US Airways and Southwest Airlines, where the star notation indicates that both clustering coefficients have been corrected so that they are difference and trend stationary, before performing the regressions (see Section \ref{sec:test}): this is common practice in applied econometrics and avoids potential spurious results from regressing one nonstationary series on another. The figures suggest that there is little clustering of order four once we have corrected for the presence of usual clustering of order three.}
    \label{fig:correlation_correction}
\end{figure}

\begin{figure}\centering
    \includegraphics[scale=0.45]{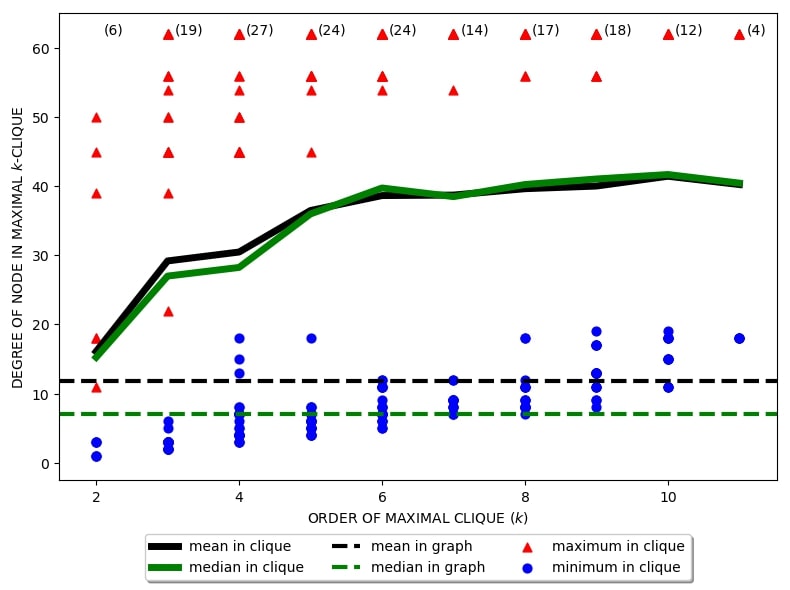}
    \caption{The mean (black line), median (green line), minimum (blue circle) and maximum (red triangle) degree of nodes that belong to maximal cliques of order $k$, in Southwest's 2013Q4 network. For comparison, we plot the mean (black dashed line) and median (green dashed line) degree of all nodes in the network. Values in parentheses are the total number of maximal cliques of order $k$ in the network. We see that airports that belong to maximal cliques are more connected, on average, than those that are not in a maximal clique, and that their degree increases in the number of airports in the clique.}\label{fig:degree_of_nodes_in_cliques_WN_2013_4}
\end{figure}

\begin{figure}\centering
              \begin{minipage}[b]{0.45\linewidth}
                \centering
                \includegraphics[scale=0.35]{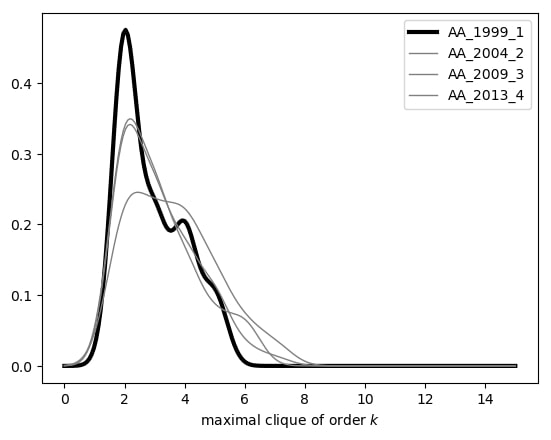}
              \end{minipage}
              \begin{minipage}[b]{0.45\linewidth}
                \centering
                \includegraphics[scale=0.35]{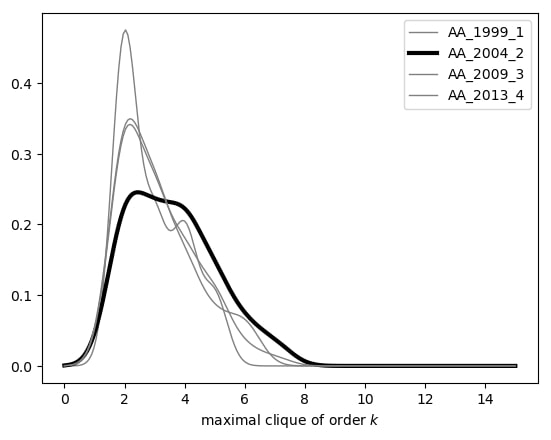}
              \end{minipage} 
              \begin{minipage}[b]{0.45\linewidth}
                \centering
                \includegraphics[scale=0.35]{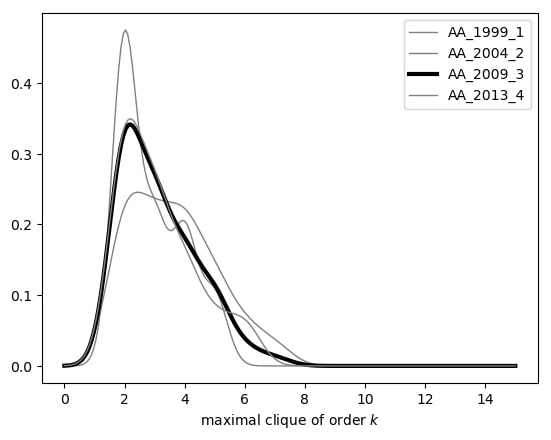}
              \end{minipage}
              \begin{minipage}[b]{0.45\linewidth}
                \centering
                \includegraphics[scale=0.35]{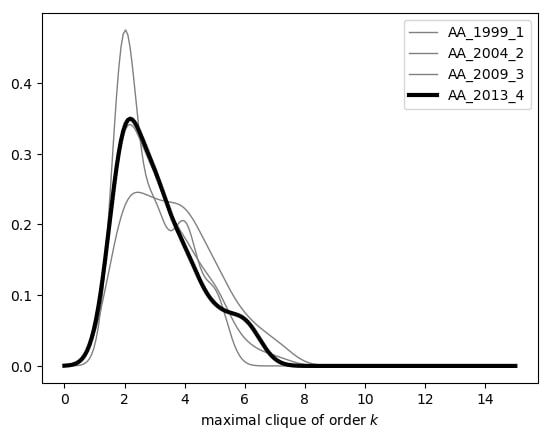}
              \end{minipage} 
              \caption{The distribution of maximal $k$-cliques in American's network, in 1999Q1, 2004Q2, 2009Q3 and 2013Q4. This shows that there is little variation over time in the number of groups of connected airports (of any size) in American's network. This is very different to what we observe for Southwest Airlines (Figure \ref{fig:kernel_plots_maximal_cliques_southwest}).}\label{fig:kernel_plots_maximal_cliques_american}
\end{figure}

\clearpage

\subsection{Computational performance of analytic formulae for $C(b)$}\label{sec:runtime}
We simulated the actual runtimes of the analytic formulae for $C(b)$ for $b=3,4,5$, on dense Erd{\H{o}}s-R{\'{e}}nyi graphs $G(n, 0.9)$, and compared these with the runtimes of a simple nested loop implementation. We are able to show that the theoretical asymptotic runtime of each of the analytic clustering formulae is lower than that of the nested loops.\footnote{The worst-case theoretical runtime of a nested loop implementation of $C(b)$ is $O(n^b)$, since there are $b$ nested loops. In a very sparse graph, the actual runtime of nested loops can be much faster, and coding shortcuts can take advantage of the fact that not every $b$-tuple needs to be considered. Directly from (\ref{eq:c_3_analytic}), (\ref{eq:c_4_analytic}) and (\ref{eq:c_5_analytic}), we can see that the numerator will dominate the asymptotic runtime of the analytic formulae. We find that $C(3)$ is $O(n^\omega)$, $C(4)$ is $O(n^{\omega+1})$, and $C(5)$ is $O(n^{\omega+2})$, where $\omega$ is the exponent of matrix multiplication, for which current implementations give $2.38 \leq \omega \leq 3$. The very fast matrix multiplication algorithms due to \cite{coppersmith_winograd90} and \cite{vassilevskawilliams14} both have $\omega \approx 2.38$, the well-known algorithm due to \cite{strassen69} has $\omega \approx 2.81$, and a na{\"i}ve algorithm has $\omega=3$.} However, the small-sample runtime is much lower when analytics are used (Figure \ref{fig:cb_runtimes}): the analytic algorithm is roughly 2,000 times faster for $C(3)$ and more than 500 times faster for $C(4)$ and $C(5)$ for the dense $G(n, p)$. While analytic runtime gains are lower for sparse $G(n, p)$, they remain very substantial, and this contributes to making these generalized clustering coefficients a practical tool for small graphs. We expect that numerical algorithms similar to those used in \cite{yin_etal18} will be more appropriate as the size of the graph increases.

\clearpage

\begin{figure}\centering
    \begin{subfigure}[b]{0.75\linewidth}
            \centering
            \includegraphics[scale=0.25]{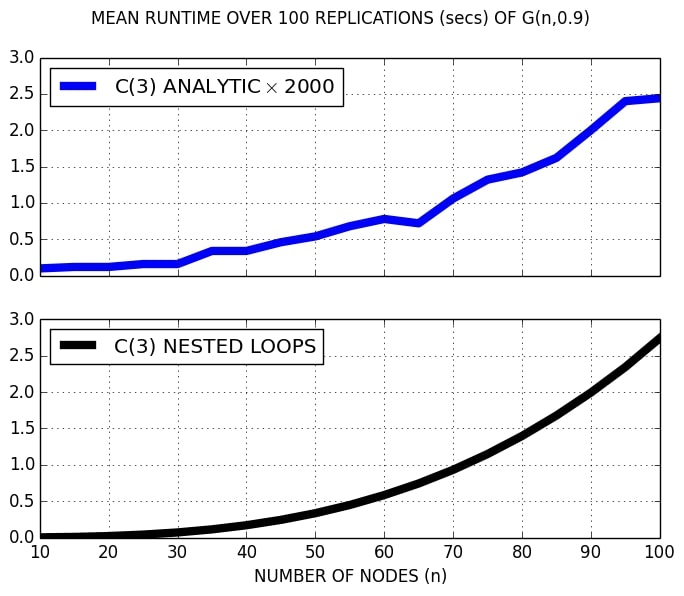}
            \caption{$C(3)$.}
            \vspace{3mm}
    \end{subfigure}
    \begin{subfigure}[b]{0.75\linewidth}
            \centering
            \includegraphics[scale=0.25]{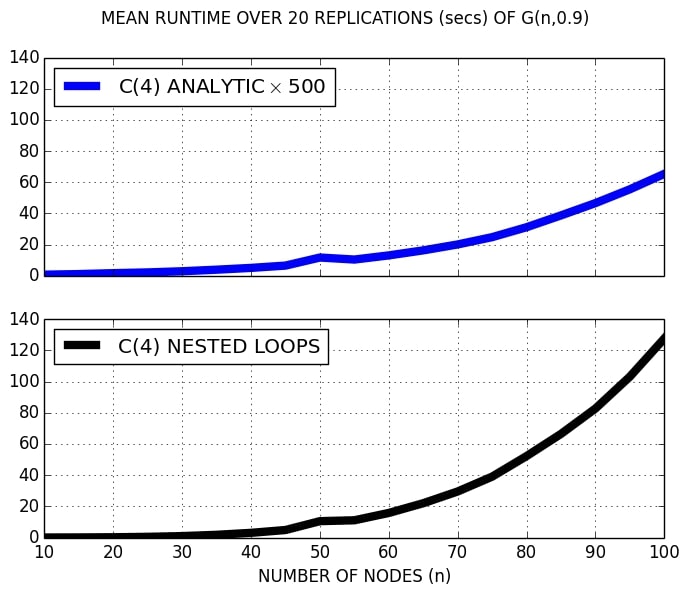}
            \caption{$C(4)$.}
            \vspace{3mm}
    \end{subfigure}
    \begin{subfigure}[b]{0.75\linewidth}
            \centering
            \includegraphics[scale=0.25]{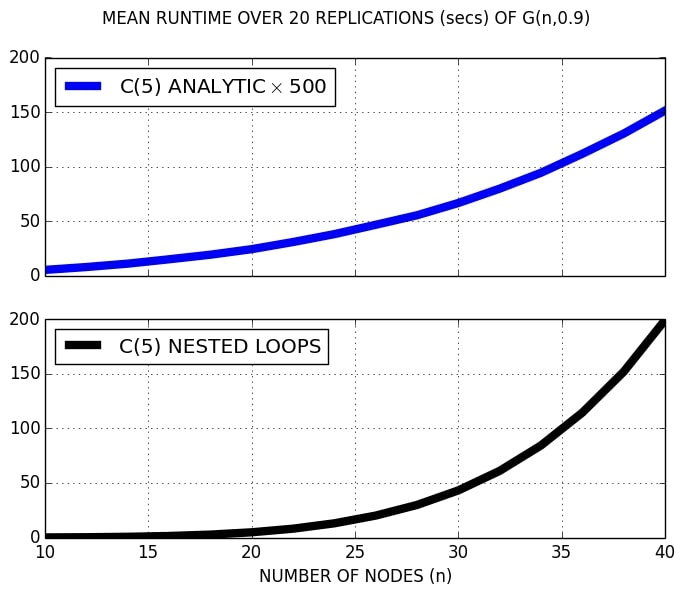}
            \caption{$C(5)$.}
            \vspace{3mm}
    \end{subfigure}
    \caption{Simulated runtimes, in seconds, of $C(3)$, $C(4)$ and $C(5)$ analytic and nested loop algorithms, computed over 100 (or 20 for $C(4)$ and $C(5)$) replications of dense Erd{\H{o}}s-R{\'{e}}nyi graphs $G(n, 0.9)$, where we only retain connected graphs. The analytic algorithm is roughly 2,000 times faster for usual clustering $C(3)$ and more than 500 times faster for generalized clustering $C(4)$ and $C(5)$.}
    \label{fig:cb_runtimes}
\end{figure}

\clearpage

\subsection{Supplementary results for the Yin-Benson-Leskovec statistic}

\begin{figure}[!htb]\centering
    \includegraphics[scale=0.45]{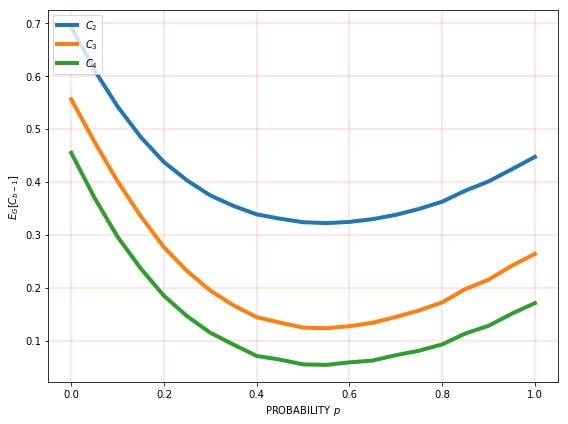}
    \caption{The simulated expected Yin-Benson-Leskovec clustering coefficient $\mathbb{E}_{G}[C_{b-1}]$ from 250 replications of a small-world graph with $n=50$ nodes, each of which has degree $2k=14$, and edge-rewiring probability $0 \leq p \leq 1$. See Figure \ref{fig:small_world_n50_k7_M200} for more details on the construction of the graph. As for $\mathbb{E}_{G}[C(b)]$ in Figure \ref{fig:small_world_n50_k7_M200}, we observe that expected clustering falls in $b$ but that it is not monotonic as $p$ increases.}\label{fig:small_world_n50_k7_M200_ybl}
\end{figure}

\begin{figure}[!htb]\centering
    \includegraphics[scale=0.45]{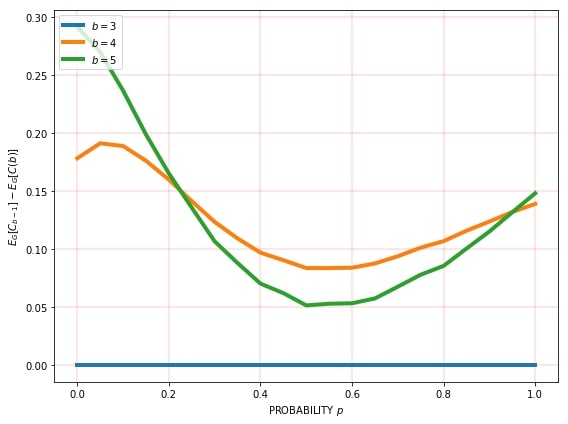}
    \caption{The simulated difference in expectation $\mathbb{E}_{G}[C_{b-1}] - \mathbb{E}_{G}[C(b)]$ between the Yin-Benson-Leskovec clustering coefficient $C_{b-1}$ and our coefficient $C(b)$ from 250 replications of the small-world graph of \cite{watts_strogatz98} with $n=50$ nodes, each of which has degree $2k=14$, and edge-rewiring probability $0 \leq p \leq 1$. See Figure \ref{fig:small_world_n50_k7_M200} for more details on the construction of the graph. We note that $C(3)=C_{2}$. We observe that the expected difference is positive and can be quite substantial, as for the Erd{\H{o}}s-R{\'{e}}nyi random graph $G(n, p)$ in Figure \ref{fig:clustering_difference}. Similarly to the simulated small-world expected clustering in Figure \ref{fig:small_world_n50_k7_M200} and Figure \ref{fig:small_world_n50_k7_M200_ybl}, the expected difference is not monotonic as $p$ increases.}\label{fig:clustering_difference_small_world}
\end{figure}

\clearpage

\begin{table}
\scriptsize
\begin{center}
\begin{tabular}{lrrccccccccccc}
Carrier & $C(3)$ & $C(3)_\mathrm{rand}$ & $C(4)$ & $C(4)_\mathrm{rand}$ & $C(5)$ & $C(5)_\mathrm{rand}$ & Connected \% & $C_3$ & $C_{3,\mathrm{rand}}$ & $C_4$ & $C_{4,\mathrm{rand}}$ & $C_3$ \% & $C_4$ \% \\ \hline
 & & & & & & & & & & & & & \\ 
AA & 0.120 & 0.060 & 0.018 & 0.000 & 0.002 & 0.000 & 45.3 & 0.101 & 0.003 & 0.075 & 0.000 & 100.0 & 5.7 \\ 
AS & 0.037 & 0.083 & 0.000 & 0.000 & 0.000 & 0.000 & 19.3 & 0.000 & 0.003 & NA & 0.000 & 95.4 & 1.7 \\ 
DL & 0.146 & 0.061 & 0.021 & 0.000 & 0.002 & 0.000 & 68.7 & 0.106 & 0.003 & 0.066 & 0.000 & 100.0 & 9.5 \\ 
FL & 0.154 & 0.108 & 0.008 & 0.001 & 0.000 & 0.000 & 62.7 & 0.038 & 0.007 & 0.000 & 0.000 & 100.0 & 10.4 \\ 
NK & 0.379 & 0.222 & 0.097 & 0.010 & 0.016 & 0.000 & 97.8 & 0.218 & 0.043 & 0.130 & 0.005 & 100.0 & 83.0 \\ 
UA & 0.346 & 0.138 & 0.122 & 0.003 & 0.034 & 0.000 & 95.6 & 0.259 & 0.017 & 0.189 & 0.001 & 100.0 & 65.6 \\ 
US & 0.115 & 0.067 & 0.011 & 0.000 & 0.000 & 0.000 & 34.0 & 0.054 & 0.003 & 0.000 & 0.000 & 100.0 & 4.7 \\ 
WN & 0.335 & 0.136 & 0.106 & 0.002 & 0.031 & 0.000 & 100.0 & 0.242 & 0.018 & 0.199 & 0.002 & 100.0 & 100.0 \\ 
\end{tabular}
\caption{Descriptive statistics for eight carrier networks in 2013Q4. The carriers are American Airlines (AA),                     Alaska Airlines (AS), Delta Air Lines (DL), AirTran Airways (FL), Spirit Airlines (NK), United Airlines (UA), US Airways (US),                     and Southwest Airlines (WN). We compare the clustering coefficients $C(3)$, $C(4)$ and $C(5)$ to realizations from Erd{\H{o}}s-R{\'{e}}nyi                     random graphs $G(n, p)$ with $n$ equal to the number of nodes in the observed network, and edge-formation probability $p$ equal to its density.                     Some of the random graphs are not connected (and Connected \% gives the percentage of connected realizations across all replications).                     Clustering coefficients are averaged over all 1000 realizations of $G(n, p)$, both connected and disconnected, as $C(3)_\mathrm{rand}$,                     $C(4)_\mathrm{rand}$ and $C(5)_\mathrm{rand}$. These columns are given in Table \ref{tab:network_descriptives}.                     We also report the Yin-Benson-Leskovec \cite{yin_etal18} clustering coefficients $C_3$ and $C_4$ (note that $C_2 = C(3)$ and this is                     not included separately). When $C_{b-1} = 0$, so that there are no $b$-cliques in the network, then there cannot be any $L(b, 1)$ lollipops and so the statistic $C_{b}$ will be undefined.                     In the empirical data, $C_4$ is not defined on the AS network. Conversely, $C(b)$ will always be defined on a connected graph.                     Both $C_3$ and $C_4$ detect higher-order clustering when it is present, and $C_3 > C_4$ for each network.                     There is also evidence that $C_3 \geq C(4)$ and $C_4 \geq C(5)$ whenever the Yin-Benson-Leskovec statistics are defined.                     The clustering coefficients $C_3$ and $C_4$ are averaged over all realizations of $G(n, p)$ for which the statistics are defined,                     and are reported as $C_{3,\mathrm{rand}}$ and $C_{4,\mathrm{rand}}$ (and $C_3$ \% and $C_4$ \% give the percentage of realizations                     for which each statistic is properly defined). We see that Yin-Benson-Leskovec clustering coefficients $C_3$ and $C_4$ are typically higher than random for all carriers,                     and that $C_4$ cannot be computed on many realizations of the sparse random graphs.}
\label{tab:network_descriptives_ybl}
\end{center}
\end{table}

\clearpage

\section*{Acknowledgements}
We are grateful to the editor Youjin Deng and two referees, whose comments greatly improved the paper, and to Chantal Roucolle and Tatiana Seregina. The usual caveat applies. This research did not receive any specific grant from funding agencies in the public, commercial, or not-for-profit sectors.\\
\\
Keywords: Airline network, Clique, Higher-order clustering, Graph theory, Subgraph.\\
\\
PACS numbers: 02.10.Ox (Combinatorics; graph theory), 89.40.Dd (Air transportation), 89.65.Gh (Economics; econophysics; financial markets; business and management), 89.75.-k (Complex systems).\\
\\
JEL classification: L14 (Transactional Relationships; Contracts and Reputation; Networks), L22 (Firm Organization and Market Structure), L93 (Air Transportation), C65 (Miscellaneous Mathematical Tools).

\clearpage

\singlespace

\bibliographystyle{plainnat}

\bibliography{ms}

\begin{thebibliography}{83}
\providecommand{\natexlab}[1]{#1}
\providecommand{\url}[1]{\texttt{#1}}
\expandafter\ifx\csname urlstyle\endcsname\relax
  \providecommand{\doi}[1]{doi: #1}\else
  \providecommand{\doi}{doi: \begingroup \urlstyle{rm}\Url}\fi

\bibitem[Agasse-Duval and Lawford(2018)]{agasse-duval_lawford18}
M.~Agasse-Duval and S.~Lawford.
\newblock Subgraphs and motifs in a dynamic airline network.
\newblock Technical Report arXiv:1807.02585, 2018.

\bibitem[Aguirregabiria and Ho(2012)]{aguirregabiria_ho12}
V.~Aguirregabiria and {C.-Y.} Ho.
\newblock A dynamic oligopoly game of the {US} airline industry: {E}stimation
  and policy experiments.
\newblock \emph{Journal of Econometrics}, 168:\penalty0 156--173, 2012.

\bibitem[Akbas et~al.(2016)Akbas, Meschke, and Wintoki]{akbas_etal16}
F.~Akbas, F.~Meschke, and M.B. Wintoki.
\newblock Director networks and informed traders.
\newblock \emph{Journal of Accounting and Economics}, 62:\penalty0 1--23, 2016.

\bibitem[Albert and Barab{\'a}si(2002)]{albert_barabasi02}
R.~Albert and {A.-L.} Barab{\'a}si.
\newblock Statistical mechanics of complex networks.
\newblock \emph{Reviews of Modern Physics}, 74:\penalty0 47--97, 2002.

\bibitem[Amaral and Ottino(2004)]{amaral_ottino04}
L.A.N. Amaral and J.M. Ottino.
\newblock Complex networks.
\newblock \emph{European Physical Journal B}, 38:\penalty0 147--162, 2004.

\bibitem[Banerjee et~al.(2013)Banerjee, Chandrasekhar, Duflo, and
  Jackson]{banerjee_etal13}
A.~Banerjee, A.G. Chandrasekhar, E.~Duflo, and M.O. Jackson.
\newblock The diffusion of microfinance.
\newblock \emph{Science}, 341:\penalty0 1236498, 2013.

\bibitem[Barab{\'a}si and Albert(1999)]{barabasi_albert99}
{A.-L.} Barab{\'a}si and R.~Albert.
\newblock Emergence of scaling in random networks.
\newblock \emph{Science}, 286:\penalty0 509--512, 1999.

\bibitem[Barrat and Weigt(2000)]{barrat_weigt00}
A.~Barrat and M.~Weigt.
\newblock On the properties of small-world network models.
\newblock \emph{European Physical Journal B}, 13:\penalty0 547--560, 2000.

\bibitem[Baumgarten et~al.(2014)Baumgarten, Malina, and
  Lange]{baumgarten_etal14}
P.~Baumgarten, R.~Malina, and A.~Lange.
\newblock The impact of hubbing concentration on flight delays within airline
  networks: {A}n empirical analysis of the {US} domestic market.
\newblock \emph{Transportation Research E}, 66:\penalty0 103--114, 2014.

\bibitem[Benson(2017)]{benson17}
A.R. Benson.
\newblock \emph{Tools for higher-order network analysis}.
\newblock PhD thesis, Stanford University, 2017.
\newblock Available from \url{http://arxiv.org/pdf/1802.06820.pdf}.

\bibitem[Benson et~al.(2016)Benson, Gleich, and Leskovec]{benson_etal16}
A.R. Benson, D.F. Gleich, and J.~Leskovec.
\newblock Higher-order organization of complex networks.
\newblock \emph{Science}, 353:\penalty0 163--166, 2016.

\bibitem[Bombelli et~al.(2020)Bombelli, Santos, and Tavasszy]{bombelli_etal20}
A.~Bombelli, B.F. Santos, and L.~Tavasszy.
\newblock Analysis of the air cargo transport network using a complex network
  theory perspective.
\newblock \emph{Transportation Research Part E}, 138:\penalty0 101959, 2020.

\bibitem[Boulet and Jouve(2008)]{boulet_jouve08}
R.~Boulet and B.~Jouve.
\newblock The lollipop graph is determined by its spectrum.
\newblock \emph{Electronic Journal of Combinatorics}, 15, 2008.

\bibitem[Bounova(2009)]{bounova09}
G.A. Bounova.
\newblock \emph{Topological evolution of networks: {C}ase studies in the {US}
  airlines and language {W}ikipedias}.
\newblock PhD thesis, Massachusetts Institute of Technology, 2009.

\bibitem[Brightwell and Winkler(1990)]{brightwell_winkler90}
G.~Brightwell and P.~Winkler.
\newblock Maximum hitting time for random walks on graphs.
\newblock \emph{Random Structures \& Algorithms}, 1:\penalty0 263--276, 1990.

\bibitem[Caldarelli et~al.(2004)Caldarelli, {Pastor-Satorras}, and
  Vespignani]{caldarelli_etal04}
G.~Caldarelli, R.~{Pastor-Satorras}, and A.~Vespignani.
\newblock Structure of cycles and local ordering in complex networks.
\newblock \emph{European Physical Journal B}, 38:\penalty0 183--186, 2004.

\bibitem[Chakraborty et~al.(2019)Chakraborty, Chowdhury, Chakraborty, Mehera,
  and Pal]{chakraborty_etal19}
M.~Chakraborty, S.~Chowdhury, J.~Chakraborty, R.~Mehera, and R.K. Pal.
\newblock Algorithms for generating all possible spanning trees of a simple
  undirected connected graph: {A}n extensive review.
\newblock \emph{Complex \& Intelligent Systems}, 5:\penalty0 265--281, 2019.

\bibitem[Chen et~al.(2020)Chen, Wang, and Jin]{chen_etal20}
Y.~Chen, J.~Wang, and F.~Jin.
\newblock Robustness of {C}hina's air transport network from 1975 to 2017.
\newblock \emph{Physica A}, 539:\penalty0 122876, 2020.

\bibitem[Cheung et~al.(2020)Cheung, Wong, and Zhang]{cheung_etal20}
T.K.Y. Cheung, C.W.H. Wong, and A.~Zhang.
\newblock The evolution of aviation network: {G}lobal airport connectivity
  index 2006--2016.
\newblock \emph{Transportation Research Part E}, 133:\penalty0 101826, 2020.

\bibitem[Ciliberto and Tamer(2009)]{ciliberto_tamer09}
F.~Ciliberto and E.~Tamer.
\newblock Market structure and multiple equilibria in airline markets.
\newblock \emph{Econometrica}, 77:\penalty0 1791--1828, 2009.

\bibitem[Ciliberto and Williams(2010)]{ciliberto_williams10}
F.~Ciliberto and J.W. Williams.
\newblock Limited access to airport facilities and market power in the airline
  industry.
\newblock \emph{Journal of Law and Economics}, 53:\penalty0 467--495, 2010.

\bibitem[Cimini et~al.(2019)Cimini, Squartini, Saracco, Garlaschelli,
  Gabrielli, and Caldarelli]{cimini_etal19}
G.~Cimini, T.~Squartini, F.~Saracco, D.~Garlaschelli, A.~Gabrielli, and
  G.~Caldarelli.
\newblock The statistical physics of real-world networks.
\newblock \emph{Nature Reviews Physics}, 1:\penalty0 58--71, 2019.

\bibitem[{Cohen-Cole} et~al.(2014){Cohen-Cole}, Kirilenko, and
  Patacchini]{cohen-cole_etal14}
E.~{Cohen-Cole}, A.~Kirilenko, and E.~Patacchini.
\newblock Trading networks and liquidity provision.
\newblock \emph{Journal of Financial Economics}, 113:\penalty0 235--251, 2014.

\bibitem[Coppersmith and Winograd(1990)]{coppersmith_winograd90}
D.~Coppersmith and S.~Winograd.
\newblock Matrix multiplication via arithmetic progressions.
\newblock \emph{Journal of Symbolic Computation}, 9:\penalty0 251--280, 1990.

\bibitem[Costa et~al.(2007)Costa, Rodrigues, Travieso, and {Villas
  Boas}]{costa_etal07}
{L. da F.} Costa, F.A. Rodrigues, G.~Travieso, and P.R. {Villas Boas}.
\newblock Characterization of complex networks: {A} survey of measurements.
\newblock \emph{Advances in Physics}, 56:\penalty0 167--242, 2007.

\bibitem[Dai et~al.(2014)Dai, Liu, and Serfes]{dai_etal14}
M.~Dai, Q.~Liu, and K.~Serfes.
\newblock Is the effect of competition on price dispersion non-monotonic?
  {E}vidence from the {U.S.} airline industry.
\newblock \emph{Review of Economics and Statistics}, 96:\penalty0 161--170,
  2014.

\bibitem[{de Paula}(2017)]{depaula17}
\'{A}. {de Paula}.
\newblock Econometrics of network models.
\newblock In B.~Honor\'{e}, A.~Pakes, M.~Piazzasi, and L.~Samuelson, editors,
  \emph{Advances in Economics and Econometrics, (Proceedings of the 11th World
  Congress of the Econometric Society)}, pages 268--323. Cambridge University
  Press, 2017.

\bibitem[{de Paula}(2020)]{depaula20}
\'{A}. {de Paula}.
\newblock Econometric models of network formation.
\newblock \emph{Annual Review of Economics}, 12, 2020.

\bibitem[Diestel(2017)]{diestel17}
R.~Diestel.
\newblock \emph{Graph Theory}.
\newblock Springer, 5th edition, 2017.

\bibitem[Du et~al.(2016)Du, Zhou, Lordan, Wang, Zhao, and Zhu]{du_etal16}
{W.-B.} Du, {X.-L.} Zhou, O.~Lordan, Z.~Wang, C.~Zhao, and {Y.-B.} Zhu.
\newblock Analysis of the {C}hinese {A}irline {N}etwork as multi-layer
  networks.
\newblock \emph{Transportation Research Part E}, 89:\penalty0 108--116, 2016.

\bibitem[{El-Khatib} et~al.(2015){El-Khatib}, Fogel, and
  Jandik]{el-khatib_etal15}
R.~{El-Khatib}, K.~Fogel, and T.~Jandik.
\newblock {CEO} network centrality and merger performance.
\newblock \emph{Journal of Financial Economics}, 116:\penalty0 349--382, 2015.

\bibitem[Faris and Felmlee(2011)]{faris_felmlee11}
R.~Faris and D.~Felmlee.
\newblock Status struggles: {N}etwork centrality and gender segregation in
  same- and cross-gender aggression.
\newblock \emph{American Sociological Review}, 76:\penalty0 48--73, 2011.

\bibitem[Fox(2008)]{fox08}
J.~Fox.
\newblock There exist graphs with super-exponential {R}amsey multiplicity
  constant.
\newblock \emph{Journal of Graph Theory}, 57:\penalty0 89--98, 2008.

\bibitem[Fronczak et~al.(2002)Fronczak, Ho{\l}yst, Jedynak, and
  Sienkiewicz]{fronczak_etal02}
A.~Fronczak, J.A. Ho{\l}yst, M.~Jedynak, and J.~Sienkiewicz.
\newblock Higher order clustering coefficients in {B}arab{\'a}si--{A}lbert
  networks.
\newblock \emph{Physica A}, 316:\penalty0 688--694, 2002.

\bibitem[Gautreau et~al.(2009)Gautreau, Barrat, and
  Barth{\'e}lemy]{gautreau_etal09}
A.~Gautreau, A.~Barrat, and M.~Barth{\'e}lemy.
\newblock Microdynamics in stationary complex networks.
\newblock \emph{PNAS}, 106:\penalty0 8847--8852, 2009.

\bibitem[Graham and {de Paula}(2020)]{graham_depaula20}
B.~Graham and \'{A}. {de Paula}, editors.
\newblock \emph{The Econometric Analysis of Network Data}.
\newblock Academic Press, 2020.

\bibitem[Guimer{\`a} and Amaral(2004)]{guimera_amaral04}
R.~Guimer{\`a} and L.A.N. Amaral.
\newblock Modeling the world-wide airport network.
\newblock \emph{European Physical Journal B}, 38:\penalty0 381--385, 2004.

\bibitem[Guimer{\`a} et~al.(2005)Guimer{\`a}, Mossa, Turtschi, and
  Amaral]{guimera_etal05}
R.~Guimer{\`a}, S.~Mossa, A.~Turtschi, and L.A.N. Amaral.
\newblock The worldwide air transportation network: {A}nomalous centrality,
  community structure, and cities' global roles.
\newblock \emph{PNAS}, 102:\penalty0 7794--7799, 2005.

\bibitem[Haemers et~al.(2008)Haemers, Liu, and Zhang]{haemers_etal08}
W.H. Haemers, X.~Liu, and Y.~Zhang.
\newblock Spectral characterizations of lollipop graphs.
\newblock \emph{Linear {A}lgebra and its {A}pplications}, 428:\penalty0
  2415--2423, 2008.

\bibitem[Hochberg et~al.(2007)Hochberg, Ljungqvist, and Lu]{hochberg_etal07}
Y.V. Hochberg, A.~Ljungqvist, and Y.~Lu.
\newblock Whom you know matters: {V}enture capital networks and investment
  performance.
\newblock \emph{Journal of Finance}, 62:\penalty0 251--301, 2007.

\bibitem[Jackson(2008)]{jackson08}
M.O. Jackson.
\newblock \emph{Social and Economic Networks}.
\newblock Princeton University Press, 2008.

\bibitem[Jackson(2014)]{jackson14}
M.O. Jackson.
\newblock Networks in the understanding of economic behaviors.
\newblock \emph{Journal of Economic Perspectives}, 28:\penalty0 3--22, 2014.

\bibitem[Jackson and Rogers(2005)]{jackson_rogers05}
M.O. Jackson and B.W. Rogers.
\newblock The economics of small worlds.
\newblock \emph{Journal of the European Economic Association}, 3:\penalty0
  617--627, 2005.

\bibitem[Jackson et~al.(2017)Jackson, Rogers, and Zenou]{jackson_etal17}
M.O. Jackson, B.W. Rogers, and Y.~Zenou.
\newblock The economic consequences of social-network structure.
\newblock \emph{Journal of Economic Literature}, 55:\penalty0 49--95, 2017.

\bibitem[Jeong et~al.(2000)Jeong, Tombor, Albert, Oltvai, and
  Barab{\'a}si]{jeong_etal00}
H.~Jeong, B.~Tombor, R.~Albert, Z.N. Oltvai, and {A.-L.} Barab{\'a}si.
\newblock The large-scale organization of metabolic networks.
\newblock \emph{Nature}, 407:\penalty0 651--654, 2000.

\bibitem[Jiang and Claramunt(2004)]{jiang_claramunt04}
B.~Jiang and C.~Claramunt.
\newblock Topological analysis of urban street networks.
\newblock \emph{Environment and Planning B: Planning and Design}, 31:\penalty0
  151--162, 2004.

\bibitem[Jungnickel(2008)]{jungnickel08}
D.~Jungnickel.
\newblock \emph{Graphs, Networks and Algorithms}.
\newblock Springer, 3rd edition, 2008.

\bibitem[Lawford(2020)]{lawford20}
S.~Lawford.
\newblock Counting five node subgraphs.
\newblock DEVI/ENAC unpublished report, 2020.

\bibitem[Lawler(1986)]{lawler86}
G.F. Lawler.
\newblock Expected hitting times for a random walk on a connected graph.
\newblock \emph{Discrete Mathematics}, 61:\penalty0 85--92, 1986.

\bibitem[Lin and Ban(2014)]{lin_ban14}
J.~Lin and Y.~Ban.
\newblock The evolving network structure of {US} airline system during
  1990-2010.
\newblock \emph{Physica A}, 410:\penalty0 302--312, 2014.

\bibitem[Lordan and Sallan(2019)]{lordan_sallan19}
O.~Lordan and J.M. Sallan.
\newblock Core and critical cities of global region airport networks.
\newblock \emph{Physica A}, 513:\penalty0 724--733, 2019.

\bibitem[Lordan et~al.(2014)Lordan, Sallan, and Simo]{lordan_etal14}
O.~Lordan, J.M. Sallan, and P.~Simo.
\newblock Study of the topology and robustness of airline route networks from
  the complex network approach: a survey and research agenda.
\newblock \emph{Journal of Transport Geography}, 37:\penalty0 112--120, 2014.

\bibitem[Malighetti et~al.(2019)Malighetti, Martini, Redondi, and
  Scotti]{malighetti_etal19}
P.~Malighetti, G.~Martini, R.~Redondi, and D.~Scotti.
\newblock Air transport networks of global integrators in the more liberalized
  {A}sian air cargo industry.
\newblock \emph{Transport Policy}, 80:\penalty0 12--23, 2019.

\bibitem[Marvel et~al.(2013)Marvel, Martin, Doering, Lusseau, and
  Newman]{marvel_etal13}
S.A. Marvel, T.~Martin, C.R. Doering, D.~Lusseau, and M.E.J. Newman.
\newblock The small-world effect is a modern phenomenon.
\newblock Technical Report arXiv:1310.2636, 2013.

\bibitem[Milo et~al.(2002)Milo, {Shen-Orr}, Itzkovitz, Kashtan, Chklovskii, and
  Alon]{milo_etal02}
R.~Milo, S.~{Shen-Orr}, S.~Itzkovitz, N.~Kashtan, D.~Chklovskii, and U.~Alon.
\newblock Network motifs: {S}imple building blocks of complex networks.
\newblock \emph{Science}, 298:\penalty0 824--827, 2002.

\bibitem[Movarraei and Shikare(2014)]{movarraei_shikare14}
N.~Movarraei and M.M. Shikare.
\newblock On the number of paths of lengths 3 and 4 in a graph.
\newblock \emph{International Journal of Applied Mathematical Research},
  3:\penalty0 178--189, 2014.

\bibitem[Newman(2000)]{newman00}
M.E.J. Newman.
\newblock Models of the small world.
\newblock \emph{Journal of Statistical Physics}, 101:\penalty0 819--841, 2000.

\bibitem[Newman(2001)]{newman01c}
M.E.J. Newman.
\newblock Clustering and preferential attachment in growing networks.
\newblock \emph{Physical Review E}, 64:\penalty0 025102, 2001.

\bibitem[Newman(2003{\natexlab{a}})]{newman03}
M.E.J. Newman.
\newblock The structure and function of complex networks.
\newblock \emph{SIAM Review}, 45:\penalty0 167--256, 2003{\natexlab{a}}.

\bibitem[Newman(2003{\natexlab{b}})]{newman03d}
M.E.J. Newman.
\newblock Properties of highly clustered networks.
\newblock \emph{Physical Review E}, 68:\penalty0 026121, 2003{\natexlab{b}}.

\bibitem[Newman(2009)]{newman09}
M.E.J. Newman.
\newblock Random graphs with clustering.
\newblock \emph{Physical Review Letters}, 103:\penalty0 058701, 2009.

\bibitem[Newman et~al.(2001)Newman, Strogatz, and Watts]{newman_etal01}
M.E.J. Newman, S.H. Strogatz, and D.J. Watts.
\newblock Random graphs with arbitrary degree distributions and their
  applications.
\newblock \emph{Physical Review E}, 64:\penalty0 026118, 2001.

\bibitem[Reggiani et~al.(2010)Reggiani, Nijkamp, and Cento]{reggiani_etal10}
A.~Reggiani, P.~Nijkamp, and A.~Cento.
\newblock Connectivity and concentration in airline networks: {A} complexity
  analysis of {L}ufthansa's network.
\newblock \emph{European Journal of Information Systems}, 119:\penalty0
  449--461, 2010.

\bibitem[Robinson and Stuart(2004)]{robinson_stuart04}
D.T. Robinson and T.E. Stuart.
\newblock Network effects in the governance of strategic alliances.
\newblock \emph{Journal of Law, Economics, \& Organization}, 23:\penalty0
  242--273, 2004.

\bibitem[Roucolle et~al.(2020{\natexlab{a}})Roucolle, Seregina, and
  Urdanoz]{roucolle_etal20}
C.~Roucolle, T.~Seregina, and M.~Urdanoz.
\newblock Measuring the development of airline networks: {C}omprehensive
  indicators.
\newblock \emph{Transportation Research Part A}, 133:\penalty0 303--324,
  2020{\natexlab{a}}.

\bibitem[Roucolle et~al.(2020{\natexlab{b}})Roucolle, Seregina, and
  Urdanoz]{roucolle_etal20b}
C.~Roucolle, T.~Seregina, and M.~Urdanoz.
\newblock Network development and excess travel time.
\newblock \emph{Transport Policy}, 94:\penalty0 139--152, 2020{\natexlab{b}}.

\bibitem[Ryczkowski et~al.(2017)Ryczkowski, Fronczak, and
  Fronczak]{ryczkowski_etal17}
T.~Ryczkowski, A.~Fronczak, and P.~Fronczak.
\newblock How transfer flights shape the structure of the airline network.
\newblock \emph{Scientific Reports}, 7:\penalty0 5630, 2017.

\bibitem[Strassen(1969)]{strassen69}
V.~Strassen.
\newblock Gaussian elimination is not optimal.
\newblock \emph{Numerische Mathematik}, 14:\penalty0 354--356, 1969.

\bibitem[Strogatz(2001)]{strogatz01}
S.H. Strogatz.
\newblock Exploring complex networks.
\newblock \emph{Nature}, 410:\penalty0 268--276, 2001.

\bibitem[{Vassilevska Williams}(2014)]{vassilevskawilliams14}
V.~{Vassilevska Williams}.
\newblock Multiplying matrices in ${O}(n^{2.373})$ time.
\newblock Mimeo (available at: \
  \url{http://people.csail.mit.edu/virgi/matrixmult-f.pdf}\ ), 2014.

\bibitem[Verma et~al.(2014)Verma, Ara{\'u}jo, and Herrmann]{verma_etal14}
T.~Verma, N.A.M. Ara{\'u}jo, and H.J. Herrmann.
\newblock Revealing the structure of the world airline network.
\newblock \emph{Scientific Reports}, 4:\penalty0 5638, 2014.

\bibitem[Watts and Strogatz(1998)]{watts_strogatz98}
D.J. Watts and S.H. Strogatz.
\newblock Collective dynamics of `small-world' networks.
\newblock \emph{Nature}, 393:\penalty0 440--442, 1998.

\bibitem[Wuellner et~al.(2010)Wuellner, Roy, and {D'Souza}]{wuellner_etal10}
D.R. Wuellner, S.~Roy, and R.M. {D'Souza}.
\newblock Resilience and rewiring of the passenger airline networks in the
  {U}nited {S}tates.
\newblock \emph{Physical Review E}, 82:\penalty0 056101, 2010.

\bibitem[Yavero\u{g}lu et~al.(2014)Yavero\u{g}lu, {Malod-Dognin}, Davis,
  Levnajic, Janjic, Karapandza, Stojmirovic, and Pr\v{z}ulj]{yaveroglu_etal14}
O.N. Yavero\u{g}lu, N.~{Malod-Dognin}, D.~Davis, Z.~Levnajic, V.~Janjic,
  R.~Karapandza, A.~Stojmirovic, and N.~Pr\v{z}ulj.
\newblock Revealing the hidden language of complex networks.
\newblock \emph{Scientific Reports}, 4:\penalty0 4547, 2014.

\bibitem[Yin et~al.(2017)Yin, Benson, Leskovec, and Gleich]{yin_etal17}
H.~Yin, A.R. Benson, J.~Leskovec, and D.F. Gleich.
\newblock Local higher-order graph clustering.
\newblock In \emph{KDD 17 (Proceedings of the 23rd ACM SIGKDD International
  Conference on Knowledge Discovery and Data Mining)}, pages 555--564, 2017.

\bibitem[Yin et~al.(2018)Yin, Benson, and Leskovec]{yin_etal18}
H.~Yin, A.R. Benson, and J.~Leskovec.
\newblock Higher-order clustering in networks.
\newblock \emph{Physical Review E}, 97:\penalty0 052306, 2018.

\bibitem[Yin et~al.(2019)Yin, Benson, and Leskovec]{yin_etal19}
H.~Yin, A.R. Benson, and J.~Leskovec.
\newblock The local closure coefficient: {A} new perspective on network
  clustering.
\newblock In \emph{WSDM 19 (Proceedings of the 12th ACM International
  Conference on Web Search and Data Mining)}, pages 303--311, 2019.

\bibitem[Zaidi(2012)]{zaidi12}
F.~Zaidi.
\newblock Small world networks and clustered small world networks with random
  connectivity.
\newblock \emph{Social Network Analysis and Mining}, 3:\penalty0 51--63, 2012.

\bibitem[Zanin and Lillo(2013)]{zanin_lillo13}
M.~Zanin and F.~Lillo.
\newblock Modelling the air transport with complex networks: {A} short review.
\newblock \emph{European Physical Journal Special Topics}, 215:\penalty0 5--21,
  2013.

\bibitem[Zhu et~al.(2014)Zhu, Tian, and Cai]{zhu_etal14}
X.~Zhu, H.~Tian, and S.~Cai.
\newblock Predicting missing links via effective paths.
\newblock \emph{Physica A}, 413:\penalty0 515--522, 2014.

\bibitem[Zhu et~al.(2018)Zhu, Tian, Chen, Wang, and Cai]{zhu_etal18}
X.~Zhu, H.~Tian, X.~Chen, W.~Wang, and S.~Cai.
\newblock Heterogeneous behavioral adoption in multiplex networks.
\newblock \emph{New Journal of Physics}, 20:\penalty0 125002, 2018.

\bibitem[Zhu et~al.(2019)Zhu, Ma, Su, Tian, Wang, and Cai]{zhu_etal19b}
X.~Zhu, J.~Ma, X.~Su, H.~Tian, W.~Wang, and S.~Cai.
\newblock Information spreading on weighted multiplex social network.
\newblock \emph{Complexity}, 2019:\penalty0 5920187, 2019.

\bibitem[Zou et~al.(2019)Zou, Donner, Marwan, Donges, and Kurths]{zou_etal19}
Y.~Zou, R.V. Donner, N.~Marwan, J.F. Donges, and J.~Kurths.
\newblock Complex network approaches to nonlinear time series analysis.
\newblock \emph{Physics Reports}, 787:\penalty0 1--97, 2019.

\end{thebibliography}

\end{document}